\documentclass[fontsize=11pt, headings=normal, headinclude=true, paper=a4, pagesize, cleardoublepage=current, BCOR=10mm, fleqn, numbers=noenddot]{scrartcl}
\pdfoutput=1
\usepackage{lmodern}

\setkomafont{sectioning}{\normalcolor\bfseries}%Serifen
\recalctypearea % Nach ändern der Schriftart von KOMA-Script ganz bequem den Satzspiegel neu berechnen lassen.
\setcapindent{0pt}
  \usepackage[english]{babel}  % mehrsprachiger Textsatz
  \usepackage[numbers,square]{natbib}
  \usepackage{mathrsfs}
  \usepackage{amsmath}
  \usepackage{amssymb}
  \usepackage{amsthm}
  \usepackage{enumitem}
  \usepackage[utf8]{inputenc}    % Passen Sie die Option an.
  \usepackage[T1]{fontenc}         % Schriftenkodierung
  \usepackage{graphicx}            % zum Einbinden von Grafiken
  \usepackage{color}		%Farben, für Kommentare
  \usepackage[english=british]{csquotes} %autostyle
  \usepackage{xcolor}

  \usepackage{tikz}

\usetikzlibrary{matrix}

\newtheorem{thm}{Theorem}[section]
\newtheorem{prop}[thm]{Proposition}
\newtheorem{lem}[thm]{Lemma}
\newtheorem{cor}[thm]{Corollary}

\theoremstyle{definition}

\newtheorem{cond}[thm]{Condition}

\newtheorem{rem}[thm]{Remark}

\renewcommand{\phi}{\varphi}
\newcommand{\eps}{\varepsilon}

\newcommand{\ud}{\mathrm{d}}
\newcommand{\uod}{{\mathrm{od}}}
\newcommand{\ue}{\mathrm{e}}
\newcommand{\ui}{\mathrm{i}}
\newcommand{\R}{\mathbb{R}}
\newcommand{\C}{\mathbb{C}}
\newcommand{\N}{\mathbb{N}}
\newcommand{\D}{\mathcal{D}}

\newcommand{\Lz}{L^2}

\newcommand{\norm}[1]{\ensuremath{\left\lVert #1 \right\rVert}}
\newcommand{\abs}[1]{\ensuremath{\left\lvert #1 \right\rvert}}

\newcommand{\hilb}{\mathscr{H}}
\newcommand{\uppar}[1]{\ensuremath{^{(#1)}}}

\title{On a Direct Description of Pseudorelativistic Nelson Hamiltonians}

\author{Julian Schmidt \thanks{Fachbereich Mathematik, Eberhard Karls Universität Tübingen, Auf der Morgenstelle 10, 72076 Tübingen, Germany. \texttt{juls@maphy.uni-tuebingen.de}}
}

\begin{document}

\maketitle  

\begin{abstract}
Abstract interior-boundary conditions (IBC's) allow for the direct description of the domain and the action of Hamiltonians for a certain class of ultraviolet-divergent models in Quantum Field Theory. The method was recently applied to models where nonrelativistic scalar particles are linearly coupled to a quantised field, the best known of which is the Nelson model. Since this approach avoids the use of ultraviolet-cutoffs, there is no need for a renormalisation procedure. Here, we extend the IBC method to pseudorelativistic scalar particles that interact with a real bosonic field. We construct the Hamiltonians for such models via abstract boundary conditions, describing their action explicitly. In addition, we obtain a detailed characterisation of their domain and make the connection to renormalisation techniques. As an example, we apply the method to two relativistic variants of Nelson's model, which have been renormalised for the first time by J.~P.~Eckmann and A.~D.~Sloan in 1970 and 1974, respectively. 
\end{abstract}

\section{Introduction}
\label{sect:new}
In the recent article \cite{nelsontype}, J. Lampart together with the author used abstract boundary conditions to characterise the domain and the action of certain otherwise ultraviolet-divergent Hamiltonians. Those Hamiltonians describe models where nonrelativistic scalar particles (often called nucleons) are linearly coupled to a field of massive scalar bosons, the most prominent of which is the so called Nelson model (\cite{nelson1964}). To characterise the domains and to set up the Hamiltonians, an abstract variant of interior-boundary conditions (IBC's) was used. These conditions relate the wave functions of different sectors of Fock space. The IBC method allows for the direct description of the Nelson Hamiltonian $H_\infty$ without cutoff: no renormalisation procedure is needed. In this note, we will extend the method to also treat variants of Nelson's model where not only the kinematics of the field but also of the nucleons is relativistic. 

The formal Hamiltonian of the original Nelson model is the sum of the free operator of nucleons and field and an interaction term. For one nucleon, the free operator in Fourier representation reads $L=p^2+\ud \Gamma(\omega)$, and acts as a self-adjoint operator on the Hilbert space
\begin{align*}
\hilb := \Lz(\R^d) \otimes \Gamma(\Lz(\R^d)) = \bigoplus_{n=0}^\infty \Lz(\R^d) \otimes \Lz_{\mathrm{sym}}(\R^{ dn}) = \bigoplus_{n=0}^\infty \hilb \uppar n \, .
\end{align*} 
Here $p$ denotes the momentum of the nucleon and $\ud \Gamma(\omega)$ is the second quantisation of the field dispersion $\omega(k) = \sqrt{k^2+1}$, which acts on the bosonic Fock space $\Gamma(\Lz(\R^d))$. The sectors $\hilb \uppar n$ are equal to $\Lz(\R^d) \otimes \Lz_{\mathrm{sym}}(\R^{ dn})$, the subspaces of functions in $\Lz$ that are symmetric under exchange of the $k$-variables. The interaction term of the Nelson model is formally given by $a(V)+ a^*(V)$ where $V: \Lz(\R^{d}) \rightarrow \mathcal{D}'(\R^{d}\times \R^d)$ is a (formal) operator (for more details on these generalised creation and annihilation operators see, e.g., \cite[App. B]{GrWu16}). The operator $V$ acts as $(V \psi)(p,k) := v(k) \psi(p+k)$ with $v \in \Lz_{\mathrm{loc}}(\R^d)$ called the \textit{form factor}. In the Nelson model we have $v=\omega^{-1/2}$. That means that $v$ is not square integrable at infinity and therefore $a^*(V)$ is ill-defined as an operator into $\hilb$. 

The interaction in the Nelson model can be understood to be a coupling of the form $\int \Psi^{+}(x) (\phi^{+}(x)+\phi^{-}(x)) \Psi^{-}(x) \, \ud x$ where $\Psi^{-}(x)$ is the nonrelativistic complex scalar nucleon field, $\Psi^{+}(x)$ its adjoint and 
\begin{align*}
\phi^{+}(x)+\phi^{-}(x) = \int \omega(k)^{-1/2}(\ue^{\ui k \cdot x} a(k) + \ue^{-\ui k \cdot x} a^*(k) ) \, \ud k
\end{align*}
is the real bosonic field operator with form factor $v(k) = \omega(k)^{-1/2} = (k^2+ 1)^{-1/4}$. In trying to adapt this expression to include nucleons with relativistic kinematics, two different choices have been made:
\begin{itemize}
\item Eckmann \cite{eckmann1970} took $\Psi^{\pm}(x)$ to be, analogously to $\phi^\pm(x)$, the annihilation and creation part of a relativistic scalar nucelon field. The nucleons are assumed to have dispersion relation $\Theta(p)=\sqrt{p^2+\mu^2}$, where $\mu \geq 0$ is the nucleon mass. With this choice, the operators $\Psi^{\pm}(x)$ feature an additional factor $\Theta(p)^{-1/2}$ when compared to the Nelson model. For this interaction operator, the number of nucleons is still conserved, and thus restricting the investigation to a fixed number of nucleons is convenient. For one particle, the interaction in Fourier representation is still of the form $a(V)+a^*(V)$ but the form factor $v$ now becomes the function $v \in  \Lz_{\mathrm{loc}}(\R^{d} \times \R^d)$ given by $v_p(k)=\Theta(p)^{-1/2} \Theta(p+k)^{-1/2}  \omega(k)^{-1/2}$. This function is not in $\Lz(\R_k^d)$ for any $p \in \R^d$ if $d \geq 2$. The dependence of the form factor on $p$ is one major difference between Eckmann's model and the original Nelson model. 
However, the form factor of the relativistic model at hand is more regular in $k$ for $\mu>0$: it holds that $\Theta(p)^{-1/2} \Theta(p+k)^{-1/2} \leq ( \mu \abs{k})^{-1/2}$ pointwise on $\R^{d} \times \R^d$.
\item Gross \cite{gross1973} also assumed relativistic kinematics of the form $\Theta(p)=\sqrt{p^2+\mu^2}$ for the nucleons (resulting in $L=\Theta(p)+\ud \Gamma(\omega)$) but kept the operators $\Psi^{\pm}(x)$ as they were in the Nelson model: just the creation and annihilation operators for the nucleons, without any additional factors. This implies that the form factor $v_p(k)=v(k) = \omega(k)^{-1/2}$ is independent of $p \in \R^d$. It is however more singular than the one chosen by Eckmann. For the IBC method to work, one needs at least that $a(V)L^{-1}$ is continuous. Therefore, in this model, one has to restrict to $d=2$. On the other hand in Gross' model we can treat also the case $\mu =0$.
\end{itemize}
%None of the models is relativistic, because
Compared to a full Yukawa-type coupling of a complex and a real scalar field, the pair creation and pair annihilation terms have been dropped in both of these models. Models of the above type have been called \textit{polarisation-free Yukawa interaction} (\cite{alb1973}), \textit{spinless Yukawa model} (\cite{DeckertPizzo14}), or as having a \textit{persistent vacuum} (\cite{froehlich1974,eckmann1970}). Note that also the interaction of the Pauli-Fierz Hamiltonian is of the form $a(V)+a^*(V)$, when the pair creation and annihilation terms are dropped. In this case however, $v_p(k,\lambda) = e_\lambda(k) \cdot (p+k) \omega(k)^{-1/2}$ is not only singular in $k$ but also even more so in $p$. 

We will later assume that $v_p(k)$ is uniformly bounded by $\abs{k}^{-\alpha}$ for some $\alpha \in [0, d/2)$, as in \cite{nelsontype}. Such form factors do not exhibit infrared-problems, because they are in $\Lz_{\mathrm{loc}}(\R^d)$. There is however an ultraviolet-problem present due to the fact that these form factors are not necessarily square integrable at infinity and thus not in $\Lz(\R^d)$. In order to make sense of the Hamiltonian
\begin{align}
\label{eq:formalexham}
H = \Theta(p) +\ud \Gamma (\omega) + a(V) +a^*(V) \, ,
\end{align}
one would multiply the form factor $v_p(k)$ by a momentum cutoff $\chi_\Lambda(k)$ for some $\Lambda < \infty$ where $\chi_\Lambda$ denotes the characteristic function of the ball of radius $\Lambda$ in $\R^d$. The resulting operator $H_\Lambda$ is self-adjoint on the domain of the free operator $L= \Theta(p) +\ud \Gamma (\omega)$. Renormalisation would amount to finding a sequence $E_\Lambda$ such that $H_\Lambda + E_\Lambda$ converges to a self-adjoint operator $H_\infty$ in some generalised sense. Note that if $v_p$ depends on $p$, then in general also $E_\Lambda$ does. That is, $E_\Lambda(p)$ is an operator on $\Lz(\R^d)$ that effectively alters the dispersion of the nucleons, already for finite $\Lambda$. We will refer to it as a (renormalisation) counter term.
%In Eckmann's model it turns out that it is a bounded operator.
In 1970, Eckmann showed that the first model can be renormalised in this sense with $H_\Lambda+E_\Lambda$ converging in norm resolvent sense. He used a reordering of the resolvent of $H_\Lambda$ which is originally due to Hepp \cite{Heppskript}. Sloan \cite{sloan1974} showed strong resolvent convergence for the model considered by Gross in $d=2$. Fröhlich investigated the infrared behaviour of both models in \cite{froehlich1974} and Albeverio \cite{alb1973} worked on scattering theory for Eckmann's model and a related one where $E_\Lambda$ is replaced by a different operator $E'_\Lambda$. In \cite{DissAndi}, Wünsch, Schach M\o{}ller and Griesemer applied Eckmann's method to Gross' model in $d=2$ in order to show that the domain of the renormalised operator $D(H_\infty)$ is contained in $D(L^\eta)$ for all $0\leq \eta < 1/2$. 

Interior-boundary conditions were introduced in \cite{TeTu15} and it was suggested that they could be used to directly define otherwise UV divergent models of mathematical QFT. Similar boundary conditions relating different sectors of Fock space have been used several times in the past, see e.g. \cite{Mosh51}, \cite{thomas1984} and \cite{Yafaev1992}. However, they have never been applied to models on the full Fock space until \cite{IBCpaper}, where a nonrelativistic model in three dimensions with a static source was investigated.

In this note we will show that the abstract IBC method of \cite{nelsontype} can be applied to Eckmann's (in $d=3$) and to Gross' model (in $d=2$). This will allow for the direct description of $H_\infty$ as a self-adjoint operator on $\hilb$. The action of $H_\infty$ and the characterisation of its domain $D(H_\infty)$ will be given in terms of abstract boundary conditions. As a Corollary we will see that $D(\abs{H_\infty}^{1/2})\subset D(L^\eta) $ for all $\eta \in [0,1/2)$ but $ D(\abs{H_\infty}^{1/2})  \cap D(L^{1/2})= \lbrace 0 \rbrace$. In Section~\ref{sect:outlook}, we will also sketch the construction for the case of massless bosons in Eckmann's model.

In both models dicussed so far, the counter terms $E_\Lambda$ diverges for fixed $p \in \R^d$ logarithmically when $\Lambda \rightarrow \infty$, exactly as in the original Nelson model. With the method applied in \cite{nelsontype} and in the present note, slightly more singular interactions can be treated (depending on various parameters and in a way to be made precise below). Recently it was shown in \cite{La18} that the IBC approach, if modified in a suitable way, also allows for the definition of a more singular model. In this nonrelativistic model, the divergence of the renormalisation constant is linear in $\Lambda$ and most importantly, a renormalisation procedure has not been worked out before.

Let us briefly sketch the definition of the Hamiltonian. Under the assumptions we will make on $V$, $\Theta$ and $\omega$, the annihilation operator $a(V)$ is an operator which maps $D(L)$ into the Hilbert space $\hilb$. This implies that the operator $G:=- (a(V)L^{-1})^*$, which maps $\hilb \uppar n$ into $\hilb \uppar {n+1}$, is continuous on $\hilb$. Then we show that $(1-G)$ is invertible and with its help define the domain of our Hamiltonian $D(H) = \lbrace \psi \in \hilb \vert (1-G) \psi \in D(L) \rbrace$. The condition $(1-G)\psi \in D(L)$ is the abstract variant of the interior-boundary condition, it states that elements in the domain of $H$ consist of a regular part $(1-G) \psi$ and a singular part $G \psi$ which is completely determined by the wave function one sector below. On $D(H)$ one can define the self-adjoint and non-negative operator $(1-G)^* L (1-G)$. 

The main task in the construction is to extend the action of the annihilation operator in a suitable way to the domain $D(H)$, i.e., to define a properly regularised symmetric operator $(T,D(H))$ which can replace the ill-defined operator $a(V) G$. Then we define, using Kato-Rellich, $H:= (1-G)^* L (1-G) + T$ to be the direct description, the correct Hamiltonian for the model. How is this related to the formal action, containing annihilation and creation operators, that we want to implement? Recall the definition of $G$ and $G^*$, respectively, which formally yield $H=L+a^*(V)+a(V)-a(V)G+T$. So the action of $H$ is in fact equal to the desired formal action \textit{up to the addition of $T-a(V)G$.} Because $T$ is a regularised version of the annihilation operator on the range of $G$, this additional term is nothing than the ill-defined part of $a(V)G$. We can relate this argument to renormalisation by introducing a UV-cutoff, thereby replacing $G$ by $G_\Lambda=-L^{-1} a^*(V_\Lambda)$. If the extension $T$ is chosen appropriately, then in fact $T_\Lambda-a(V_\Lambda)G_\Lambda =  E_\Lambda$, where $E_\Lambda$ is the standard renormalisation counter term. Recall that the usual cutoff Hamiltonian is equal to $H_\Lambda = L + a^*(V_\Lambda) + a(V_\Lambda)$. Comparing this to the formula $H=L+a^*(V)+a(V)-a(V)G+T$, we can see that $H_\Lambda+E_\Lambda$ converges (we will prove norm resolvent sense) to $H$ in the limit $\Lambda \rightarrow \infty$. This shows that $H=H_\infty$ is the renormalised Hamiltonian. Because we explicitly identified the limiting Hamiltonian, instead of having to deal with dressing transformations or resolvent series, we are left with the well-posed task of proving a relative bound of $T$ with respect to $(1-G)^* L (1-G)$ in order to obtain a direct description of the desired operator.

In \cite{IBCpaper}, a slightly different approach involving the adjoint $L_0^*$ of the operator $L_0:=L\vert_{\ker(a(V))}$ and an extension $A$ of $a(V)$ was used. The operator $A$ is added to $L_0^*$ and their sum is then restricted to a certain subspace of $D(L_0^*)$ which makes it a self-adjoint operator. In this article the case where $v(k)=(2 \pi)^{-d/2}$ is the Fourier transform of a delta distribution was considered. There it is particularly easy to see that $\ker(a(V)) \cap D(L)$ is dense in $\hilb$, such that $L_0$ is densely defined and symmetric. We expect this to be true whenever $v \notin \Lz$.  It turns out that in this case $G$ maps into $\ker( L_0^*)$ and one can rewrite the Hamiltonian $(1-G)^* L (1-G) + T$ in the form $ L_0^* + a(V)(1-G)+ T$. In this way, the connection of the two the approaches is clearly visible, for $a(V)(1-G)+ T$ is just one particular decomposition of $A$. For general form factors $v$ however, the denseness of the kernel of $a(V)$ is not immediately obvious and has to be proved. See \cite[Lem. 2.2]{nelsontype} for the case where $v$ depends on $k$ only. We will not extend these results to the general case where $v=v_p(k)$ but work with the form $H= (1-G)^* L (1-G) + T$ of the Hamiltonian.

The construction sketched above is in some respect analogous to the one used in setting up zero-range Hamiltonians and the technical tools employed here are in fact inspired by previous works on many-body point interactions, in particular \cite{Co_etal15} and \cite{MoSe17}.

In the general case we consider a system of $M$ nucleons such that the Hilbert space is given by $\hilb=\Lz(\R^{dM}) \otimes \Gamma(\Lz(\R^d))$ and the free operator becomes $L=\sum_{i=1}^M \Theta(p_i)+\ud \Gamma(\omega)$. The general coupling operator $V$ is of the form
\begin{align}
\label{eq:couplingopdef}
V \phi(P,k) = \sum_{i=1}^M V^i \phi(P,k) = \sum_{i=1}^M v^i_{p_i}(k) \phi(P+e_i k , k)
\end{align}
Here $P=(p_1, \dots, p_M)$ and $e_i$ denotes the inclusion of the $i$-th component into $\R^{M d}$. We have absorbed the common coupling constant $g$ of \cite{nelsontype} into the form factors. Since we do not assume any statistics for the nucleons, different particles could couple differently to the field and consequently the form factors would not be the same. It may however be helpful to think of them as being of the form $v^i_{p}= g_i v_p$ with $g_i \in \C$. As will be discussed in the upcoming work \cite{timeasymmetry}, different phases of the coupling constants $g_i$ can be interpreted as complex charges and the Hamiltonians then fail to be invariant under time reversal.

\section{Assumptions and Theorems}
Let $d \in \N$ denote the dimension of the physical space and let $M \in \N$ be the number of nucleons. Let $\alpha \in [0, d/2)$, $\gamma >0$ and $0 < \beta \leq \gamma$ be real constants. Set $\D:=d-2 \alpha - \gamma$. In order to define the Hamiltonian, we will make the following three assumptions.

\begin{cond}
\label{Hdefcond}  \
 
\begin{enumerate}[label=\alph*),ref=\ref{Hdefcond}~\alph*)]
\item \label{cond01} Let $\Theta, \omega \in L^1_{\mathrm{loc}}(\R^d,\R_{\geq 0})$ and $v^i_p \in \Lz_{\mathrm{loc}}(\R^d)$ for all $p \in \R^d$ and all $1 \leq i \leq M$. Assume that $v^i_{p-k}(k)=v^i_p(-k)$, and that there is a constant $c>0$ such that $\abs{ {v}^i_{p}(k)}\leq c \abs{k}^{-\alpha}$ for all $p \in \R^d$ and any $1\leq i \leq M$. In addition assume the bounds $\abs{\Theta(p)} \geq \abs{p}^\gamma$ and $\omega(k)\geq (1+k^2)^{\beta/2}$.
\item \label{cond02} For any $\eps>0$ there is a constant $C>0$ such that
\begin{align*}
 \int_{\R^d} \frac{ \abs{v_{p}(-k)}^2 \abs{\Theta(k)-\Theta(p-k)}}{(\Theta(p-k)+\omega(k))(\Theta(k)+\omega(k))} \, \ud k \leq C \left(\abs{p}^{\D+\gamma \eps}+1\right) \, .
\end{align*}
for all $p \in \R^d$.
\item \label{cond2} We have $0\leq \D < \frac{\gamma \beta^2}{\beta^2+2 \gamma^2}$.
\end{enumerate}

\end{cond}

Condition~\ref{cond01} is a global condition that will be assumed throughout the paper. When dealing with renormalisation, the parameter $\D:=d-2 \alpha - \gamma$ will basically measure the dependence of $E_\Lambda$ on $\Lambda$, with $\D=0$ corresponding to $E_\Lambda \sim \log \Lambda$.

The Condition~\ref{cond02} is concerned with the only part of the method, for which scaling is not sufficient. This is the definition of the diagonal part of the $T$-operator. For the two models that have been discussed above, we have $\gamma=1$. 

Condition~\ref{cond2} is the generalisation of the Condition~1.1~(2) of~\cite{nelsontype} to the case $\gamma \neq 2$. The upper bound ensures that $TG$ is a well defined operator while the lower bound implies that $E_\Lambda$ diverges (pointwise). This excludes the more regular cases where $a(V)$ is defined on the form domain of the free operator, see \cite[Sect. 2]{nelsontype}

In the literature on renormalisation, two different choices for the sequences of renormalisation counter terms have been made, resulting in different limiting Hamiltonians $H_\infty$. In our setting, this will be reflected in the fact that the extension of the annihilation operator, the $T$-operator, comes in one of two variants. They will be defined later in~\eqref{eq:Tdef} below. One will be denoted as variant $T^{\nu=1}$ and the other one as $\nu=2$. For this reason, we will state the main theorem also for two different operators $H^\nu$.
\begin{thm}
\label{thm:main}
Assume the Conditions~\ref{Hdefcond}. Then the operator $G:= -(a(V) L^{-1})^*$ is continuous and the domain $D(H) := \lbrace \psi \in \hilb \vert (1-G) \psi \in D(L) \rbrace$ is dense in $\hilb$. The operator $T^\nu$ -- defined in~\eqref{eq:Tdef} for $\nu=1,2$ -- is symmetric on $D(H)$ and 
\begin{align}
H^\nu:=(1-G)^* L (1-G) + T^\nu
\end{align}
is self-adjoint and bounded from below on $D(H)$.
\end{thm}

We will prove that the models obtained by renormalisation techniques are in fact equal to our Hamiltonian $H$. As stated above, we will give two variants, in order to include both choices of the renormalisation counter term. For the convergence of the renormalised Hamiltonians to be uniform we need another assumption.

\begin{cond}
\
 \label{cond1} For any $\eps>0$ there is a positive function $F \in C_0[0,\infty)$ such that
\begin{align*}
 \int_{\R^d} \frac{(1-\chi_\Lambda(k)) \abs{v_{p}(-k)}^2 \abs{\Theta(k)-\Theta(p-k)}}{(\Theta(p-k)+\omega(k))(\Theta(k)+\omega(k))} \, \ud k \leq F(\Lambda) \left(\abs{p}^{\D+\gamma \eps}+1\right) \, .
\end{align*} 
\end{cond}

Note that this Condition~\ref{cond1} is stronger than Condition~\ref{cond02}, the latter follows from this one by setting $C:=F(0)$.

\begin{prop}
\label{prop:renorm}
Assume Conditions~\ref{Hdefcond} and \ref{cond1} and let the counter term be defined in one of two different ways:
\begin{align}
E^\nu_\Lambda(P) := \begin{cases}  \sum_{i=1}^M \int_{B_\Lambda} \abs{v^i_{p_i-k}(k)}^2 (\Theta(k)+\omega(k))^{-1} \ud k  & \nu =1 \\
\sum_{i=1}^M \int_{B_\Lambda} \abs{v^i_{p_i-k}(k)}^2 (\Theta(p_i-k)+\omega(k))^{-1}  \ud k & \nu =2  \, . \end{cases} 
\end{align} 
Let $(H_\Lambda,D(L))$ be the Hamiltonian which is given by the formal expression \eqref{eq:formalexham}, where the form factors $v^i_p$ are replaced by $v^i_p\chi_\Lambda \in \Lz(\R^d)$. Then $H_\Lambda+E_\Lambda^\nu \rightarrow H^\nu$ in norm resolvent sense.
\end{prop}

For a discussion of how to choose $T^\nu$ and $E_\Lambda^\nu$, see Remark~\ref{rem:choice}. The Theorem~\ref{prop:renorm} is a slight improvement when compared to \cite[Thm. 1.4]{nelsontype}, where only strong resolvent convergence was proved. Note that the equality $H=H_\infty$ easily follows from the weaker result because the limit is unique. However, we find that it is more satisfactory to prove convergence in norm directly by using the IBC method.

The Condition~\ref{cond3} is necessary in order to prove that intersections of the form $D(\abs{H}^{1/2}) \cap D(L^\eta)$ only contain the zero vector for suitable $\eta >0$. If we assume Condition~\ref{cond3}, we suppose that $v_p$, $\omega$, and $\Theta$ behave essentially like powers of the distance, while in general we only assume an upper bound on $v_p$ and lower bounds for $\Theta$ and $\omega$. 

\begin{cond}
 \label{cond3} For any $R>0$ there exist constants $C',C >0$ such that $\Theta(q-p) \leq C (\abs{q}^{\gamma} + 1)$ and  $v_{p}(k) \geq C' (\abs{k}^{\alpha} +1)^{-1}$ for all $p \in \R^d$ with $\abs p < R$. Furthermore, there exists a constant $\tilde C >0$, such that $\omega(k) \leq \tilde{C} (\abs{k}^\gamma+1)$.
\end{cond}

The Proposition~\ref{prop:reg} gives quite strong results when compared to \cite[Thm. 4.2]{nelsontype} but only because the Condition~\ref{cond3} is more restrictive. All concrete examples we have in mind fulfill these conditions.

\begin{prop}
\label{prop:reg}
Assume Conditions~\ref{Hdefcond} and \ref{cond3}. Then for any $\nu \in \{1,2\}$ and all $\eta \in [0, \frac{\gamma-\D}{2 \gamma})$ it holds that $D(\abs{H^\nu}^{1/2}) \subset D(L^\eta)$. If however $\eta \geq \frac{\gamma-\D}{2 \gamma}$ then it holds that $D(\abs{H^\nu}^{1/2}) \cap D(L^\eta) = \{ 0 \}$.
\end{prop}

In Section~\ref{sect:construct} we will construct the Hamiltonian in the general setting and prove Theorem~\ref{thm:main} and the Proposition~\ref{prop:renorm} for $\gamma=\beta$ and Proposition~\ref{prop:reg}. The proof of Theorem~\ref{thm:main} and Proposition~\ref{prop:renorm} in the general case $\beta < \gamma$ will be given in the Appendix~\ref{sect:appendix}.

In Section~\ref{sect:examples} we will apply the results we have obtained to the two models that have been discussed in the introduction. In the end we will prove the following Corollary:
\begin{cor}
\label{corol:nelsontypemodels}
Let $\omega(k) = \sqrt{k^2 +1}$ and $\Theta(p) = \sqrt{p^2+\mu^2}$. 
\begin{itemize}
\item If $d=3$, $v_p(k) = \Theta(p)^{-1/2}  \Theta(p+k)^{-1/2}  \omega(k)^{-1/2}$ and $\mu >0$, then the renormalised operator of Eckmann \cite{eckmann1970} is equal to $H^{\nu=2}$.
\item If $d=2$, $v_p(k) = \omega(k)^{-1/2}$ and $\mu \geq 0$ then the renormalised operator for Gross' model that has been obtained in \cite{DissAndi} is equal to $H^{\nu=1}$.
\end{itemize} 
It holds that $D(\abs{H^\nu}^{1/2}) \subset D(L^\eta)$ for any $\eta <1/2$. If $\eta \in [1/2,1]$ then $D(\abs{H^\nu}^{1/2}) \cap D(L^\eta) = \{ 0 \}$ in both models.
\end{cor}

In \cite{DissAndi}, Wünsch proved that the \textit{operator domain} $D(H^{\nu=1})$ is contained in $D(L^\eta)$ for all $\eta <1/2$ in the renormalised model of Gross and Sloan. The corresponding statement for the form domain as well as its analogue for Eckmann's model seem to be new. For both models, this is apparently also the first proof of the converse -- the fact that in both models $D(\abs{H^\nu}^{1/2}) \cap D(L^{1/2}) = \{ 0 \}$.

\section{Construction of the Hamiltonian}
\label{sect:construct}
In the whole Section, the global Condition~\ref{cond01} is assumed to hold. Because our goal is to apply the results of this section to models with $\Theta(p) = \sqrt{p^2+\mu^2}$ and $\omega(k)=\sqrt{k^2+1}$, we will pay special attention to the case of $\beta = \gamma$ where Condition~\ref{cond2} reduces to $0\leq \D < \gamma/3$. Some issues concerning the general case of $0<\beta<\gamma$ will be treated only in the Appendix~\ref{sect:appendix}. 

\subsection{The domain of the Hamiltonian}
\label{subsect:extended}
We start with a technical lemma that will turn out to be very useful later on. The proof can be found in the Appendix~\ref{sect:appendix}. We will always denote the characteristic function of a ball of radius $\Lambda$ in $\R^d$ by $\chi_\Lambda$.
\begin{lem}
\label{lem:scaling}
Let $\Lambda, \Omega \geq 0$. For any $\gamma,r,\beta>0$ and $\nu,\sigma\geq 0$ such that $d \in (\nu +\sigma ,\nu+\sigma+ r \gamma)$ there exists a $\delta_0>0$ and a constant $C>0$ such that for any $0 \leq \delta < \delta_0$ it holds that
\begin{align*}
\int_{\R^d} (1-\chi_\Lambda(k)) \frac{\abs{k}^{-\nu} \abs{p-k}^{-\sigma}}{(\abs{p-k}^\gamma + \abs{k}^\beta + \Omega )^r} \, \ud k \leq C \, \Omega^{-r+(d-\nu-\sigma)/\gamma+ \delta_\Lambda}  \, \Lambda^{- \beta \delta_\Lambda}
\end{align*}
for all $p \in \R^d$. The function $\delta_\Lambda$ is defined as $\delta_\Lambda:=\delta \cdot (1-\chi_{[0,1]}(\Lambda))$.
\end{lem}

The action of the free operator on the $n$-boson sector is given by multiplication with the function
\begin{align}
L(P,K) := \sum_{i=1}^M \Theta(p) + \sum_{j=1}^n \omega(k_j) :=\Theta(P) + \Omega(K) \, ,
\end{align}
where we make use of the notation $\sum_{j \in J} \omega(q_j) = \Omega(Q)$. We can now generalise \cite[Prop.~3.1]{nelsontype} and prove that, for $0\leq \D < \gamma$, the operator $G=-( a(V) L^{-1})^* = - L^{-1} a^*(V)$ maps into $D(L^\eta)$ for some $0\leq \eta < \frac12 - \frac{\D}{2 \gamma} \leq \frac{1}{2} $.

\begin{prop}
\label{prop:generalboundong}
Define the affine transformation $u(s):=\frac{\beta}{\gamma} s - \frac{\D}{\gamma}$ and let $s\geq0$ be such that $u(s)<1$.
Then for all $0\leq \eta < \frac{1+u(s)-s}2$ the operator $G$ is bounded from $D(N^{\max(0,1-s)/2})$ to $D(L^\eta)$ and $G_\Lambda \rightarrow G$ in this norm of continuous operators. 
\end{prop}
\begin{proof}
We will prove a bound of the form $\norm{L^\eta (G-G_\Lambda) \psi} \leq f(\Lambda) \norm{N^{\max(0,1-s)/2}\psi}$ for a continuous function $f$ on $[0, \infty)$ which tends to zero as $\Lambda \rightarrow \infty$. This proves convergence. Boundedness follows by setting $\Lambda=0$. We write $V$ also for the variant of the interaction operator that acts on the $n$-th sector, i.e. $V \psi \uppar n= \sqrt{n+1} \mathrm{Sym}((V \otimes \mathbf{1}_{n})\psi \uppar n)$, where $V$ acts on $\Lz(\R^{d M})$. Sector-wise, the action of $G-G_\Lambda$ is given by
 \begin{align*}
(G-G_\Lambda) \psi\uppar{n}(P,K)
&
= - \sum_{i=1}^M (L^{-1} \mathrm{Sym}((V^i-V^i_\Lambda) \psi \uppar n))(P,K) 
\\
&
= \frac{-1}{\sqrt{n+1}} \sum_{i=1}^M \sum_{j=1}^{n+1} \frac{(1-\chi_\Lambda(k_j)) v^i_{p_i}(k_j) \psi \uppar n(P+e_i k_j,\hat{K}_j)}{L(P,K)}.
 \end{align*}
Here $\hat{K}_j$ denotes the variables $K$ with the $j$-th entry omitted. We will define $\xi_\Lambda(k_j):=1-\chi_\Lambda(k_j)$. Observe that it is sufficient to estimate the norm of $ L^\eta  (G-G_\Lambda) \psi \uppar n$ by the sum over the norms of $\kappa_i \psi \uppar n := L^{-(1-\eta)} \mathrm{Sym}((V^i-V^i_\Lambda) \psi \uppar n)$. To do so, we use the finite dimensional Cauchy-Schwarz inequality and obtain
  \begin{align*}
\abs{\kappa_i \psi \uppar n (P,K)}^2
&
\leq
 (n+1)^{-1} \sum_{j,\mu=1}^{n+1}  \frac{\abs{\xi_\Lambda(k_j) v^i_{p_i}(k_j)}^2 \abs{\psi \uppar n(P+e_i k_j,\hat{K}_j)}^2}{L(P,K)^{2(1-\eta)} \omega(k_j)^s } \omega(k_\mu)^s   \, .
 \end{align*}
Using the inequality 
\begin{align}
\label{eq:sum}
\sum_{i=1}^n \omega(k_i)^s \leq {n}^{\max(0,1-s)} \Omega(K)^{s} \, ,
\end{align} 
we can bound the $\mu$-sum by $\omega(k_j)^s + {n}^{\max(0,1-s)} \Omega(\hat{K}_j)^{s}$. Then we use the assumptions $\abs{v^i_{p_i}(k)} \leq c \abs{k}^{-\alpha}$ and $\omega(k) \geq \abs{k}^{\beta}$ as well as $v^i_{p_i-k}(k) =v^i_{p_i}(-k) $ and obtain for the translated $\abs{\kappa_i \psi \uppar n (P-e_i k_j,K)}^2$ the bound 
  \begin{align}
&(n+1)^{-1} \sum_{j=1}^{n+1}  \frac{c^2 \xi_\Lambda(k_j) \abs{\psi \uppar n(P,\hat{K}_j)}^2}{  L(P-e_i k_j,K)^{2(1-\eta)} } \left( {n}^{\max(0,1-s)} \abs{k_j}^{-2\alpha-\beta s} \Omega(\hat{K}_j)^{s}   +  \abs{k_j}^{-2\alpha} \right)
\nonumber \\
&
=\mathrm{Sym}_k\left[ \frac{c^2 \xi_\Lambda(k_1) \abs{\psi \uppar n(P,\hat{K}_{1})}^2}{  L(P-e_i k_{1},K)^{2(1-\eta)} } ( {n}^{\max(0,1-s)} \abs{k_{1}}^{-2\alpha-\beta s} \Omega(\hat{K}_{1})^{s}   +  \abs{k_{1}}^{-2\alpha} )\right] \, .
\label{eq:decompofkappa}
 \end{align}
Here we have used the symmetry of $\psi$ and $L$. Now bound $L(P-e_i k_{1},K)$ from below by $\abs{p_i-k_1}^\gamma+\abs{k_1}^\beta+\Omega(\hat{K}_1)$ and recall that Condition~\ref{cond2} implies in particular $\beta >0$.  This together with $u(s)<1$ implies that the hypothesis of Lemma~\ref{lem:scaling} is fulfilled for the first term in~\eqref{eq:decompofkappa} and consequently
  \begin{align}
  \label{eq:exponentneg1}
& c^2   \int_{\R^d}   \frac{\xi_\Lambda(k_1) \Omega(\hat{K}_{1})^{s} \abs{k_{1}}^{-2\alpha-\beta s} }{  L(P-e_i k_{1},K)^{2(1-\eta)} }   \, \ud k_1 \leq C \Omega(\hat{K}_1)^{2 (\eta-1)+\frac{d-2\alpha-\beta s}{\gamma}+s+\delta_\Lambda} \Lambda^{-\beta  \delta_\Lambda}
\, ,
 \end{align}
 where $\delta_\Lambda:= \delta (1-\chi_{[0,1]}(\Lambda))$. If $\delta>0$ is small enough, then
 \begin{align*}
2 (\eta-1)+\frac{d-2\alpha-\beta s}{\gamma}+s+\delta_\Lambda = 2 \left(\eta -  \frac{1+u(s) - s}{2} \right) +\delta_\Lambda < 0 \, .
 \end{align*} 
 Because $\Omega \geq 1$, that means that we can simply estimate $\Omega^{2 (\eta -  \frac{1+u(s) - s}{2} ) +\delta_\Lambda} \leq 1$  in \eqref{eq:exponentneg1}.

The corresponding bound for the second term of~\eqref{eq:decompofkappa} follows by setting $s=0$. Because the function $u(s)-s$ is non-increasing it holds that $2 (\eta-1)+\frac{d-2\alpha}{\gamma}+\delta_\Lambda <0$ for the same choice of $\delta>0$. Integrating in the remaining variables $(P,\hat{K}_1)$ yields the claim.
\end{proof}

\begin{cor}
\label{lem:gonfockspace}
Assume $0\leq \D < \beta$. There exists  an $\eta \in (0,1/2)$ such that $G$ is a continuous operator from $\hilb$ to $D(L^\eta)$ and $G_\Lambda \rightarrow G$ in norm as operators in $\mathcal{L}(\hilb,D(L^\eta))$. In particular, if $\beta=\gamma$, for any $\eps>0$ small enough we can choose $\eta= \frac{1-\D/\gamma}{2}-\eps$.
\end{cor}
\begin{proof}
 We apply Proposition~\ref{prop:generalboundong}, distinguishing two cases. First, if $\D=0$ and $\beta=\gamma$, then $u(s)=s$ and we choose, for some $\eps>0$, $s_\eps=1-\eps$ and $\eta_\eps=\frac{1-\eps}2$.
 Proposition~\ref{prop:generalboundong} then gives the bound 
 \begin{equation*}
\norm{ L^{\frac{1-\eps}2}(G-G_\Lambda) \psi}_{\hilb \uppar{n+1}} \leq C(\Lambda) (1+n^{\eps/2}) \norm{\psi\uppar n}_{\hilb \uppar n} \, ,
 \end{equation*}
with $C(\Lambda) \rightarrow 0$ as $\Lambda \rightarrow \infty$. This shows that $G$ maps $\hilb $ to $D(L^{1/2-\eps})$ for all $0<\eps\leq \frac12$ in this case and that $G_\Lambda \rightarrow G$ in $\mathcal{L}(\hilb,D(L^\eta))$.
 
 In all other cases, we have $u(1)=(\beta-\D)/\gamma<1$ and we may choose in Proposition~\ref{prop:generalboundong} $s=1$ and any $0\leq\eta<\frac{\beta-\D}{2 \gamma}$.
\end{proof}

\begin{lem}
\label{lem:Dposnisleftinv}
Let $0\leq \D < \beta$. Then $1-G$ is invertible and there exists a constant $C>0$ such that
\begin{align}
\label{eq:Dposnisleftinv}
\norm{N \psi }_\hilb \leq C (\norm{N (1-G) \psi }_\hilb + \norm{ \psi }_\hilb ) \, .
\end{align}
\end{lem}
\begin{proof}
See \cite[Lemma 2.4]{nelsontype}.
\end{proof}

We can now define what will be the domain of our Hamiltonian. We choose $D(H) := \{\psi \in \hilb \vert (1-G) \psi \in D(L) \} = (1-G)^{-1} D(L) $. Since $a(V) L^{-1}=-G^*$ is a continuous operator on $\hilb$, the annihilation operator $a(V)$ is well defined on $D(L)$. It is however not defined on the range of $G$, because $G$ does not map into $D(L^{1/2})$. In the next section we will extend the action of $a(V)$ in a suitable way to elements of the form $G \phi$.

\subsection{The extension of the annihilation operator}
In this section we will extend the annihilation operator $a(V)$ to $D(H) = \{\psi \in \hilb \vert (1-G) \psi \in D(L) \}$. Decomposing elements $\phi \in D(H)$ as $\phi =(1-G)\phi +G\phi $, we observe that $a(V)$ is well defined on $(1-G)\phi$ but not on $G \phi$. For that reason, we have to define an operator $T$, which is a regularised version of the operator $ a(V) G$. The formal expression for the latter is given by
\begin{align}
&a(V)  G \psi \uppar n(P,\hat K_{n+1}) 
\label{eq:Tformally v1} 
\\ \nonumber
&
= \sqrt{n+1} \sum_{\ell=1}^M \int_{\R^d} \overline{v^\ell_{p_\ell-k_{n+1}}(k_{n+1}) } G \psi \uppar n (P-e_\ell k_{n+1},K) \, \ud k_{n+1} 
\\ \nonumber
&= \begin{aligned}[t] - \sum_{i,\ell=1}^M  \sum_{j=1}^{n+1} \int_{\R^d} & \overline{v^\ell_{p_\ell-k_{n+1}}(k_{n+1})} v^i_{p_i-\delta_{\ell i} k_{n+1}}(k_j)
\\
& \times \frac{  \psi \uppar n(P-e_\ell k_{n+1}+e_i k_j,\hat{K}_j)}{L(P-e_\ell k_{n+1},K)} \, \ud k_{n+1} \, .
\end{aligned}
%\\
%&=  -  \psi \uppar n(p,\hat{K}_{n+1}) \int_{\R^d}  \frac{ \abs{v_{p}(-k_{n+1})}^2}{L(p-k_{n+1},K)} \, \ud k_{n+1}  \\
%&
%- \sum_{j=1}^{n} \int_{\R^d} \overline{v_{p}(-k_{n+1})}  \frac{ v_{p-k_{n+1}}(k_j) \psi \uppar n(p-k_{n+1}+k_j,\hat{K}_j)}{L(p-k_{n+1},K)} \, \ud k_{n+1}
\end{align} 
Here $\delta_{\ell i}$ denotes the usual Kronecker-delta. The integrals in the terms where $j=k_{n+1}$ and $\ell =i$ do not converge in general. In order to obtain a well defined operator, we have to replace the integrals in these so called \textit{diagonal parts} of the sum by regularised ones. To do so we employ the assumption $v^\ell_{p-k}(k) = v^\ell_{p}(-k) $ for all $\ell$ and set
\begin{align}
\label{eq:intergalI}
& I_\ell(P,\hat{K}_{n+1})  
:= 
\int_{\R^d} \frac{ \abs{v^\ell_{p_\ell}(-k_{n+1})}^2 }{L(P- e_\ell k_{n+1},K) } - \frac{ \abs{v^\ell_{p_\ell}(-k_{n+1})}^2 }{\Theta(k_{n+1})+\omega(k_{n+1}) } \, \ud k_{n+1}
\end{align}
and
\begin{align}
\label{eq:intergalJ}
& J(p_\ell)  
:= 
\int_{\R^d}  \frac{ \abs{v^\ell_{p_\ell}(-k_{n+1})}^2}{\Theta(k_{n+1})+\omega(k_{n+1}) } - \frac{\abs{v^\ell_{p_\ell}(-k_{n+1})}^2 }{\Theta(p_\ell-k_{n+1})+\omega(k_{n+1}) } \, \ud k_{n+1} \, .
\end{align}
Then we define two variants of the \textit{diagonal} part of the operator $T$:
\begin{align}
\label{eq:Tddef}
T_\ud^\nu \phi \uppar n (P,\hat{K}_{n+1})  := \begin{cases} - \sum_{\ell=1}^M I_\ell(P,\hat{K}_{n+1}) \phi \uppar n (P,\hat{K}_{n+1}) & \nu =1 \\
 -  \sum_{\ell=1}^M (I_\ell(P,\hat{K}_{n+1})+J(p_\ell)) \phi \uppar n (P,\hat{K}_{n+1}) & \nu =2 \, .\end{cases}
\end{align}

The remaining expressions in~\eqref{eq:Tformally v1} constitute the \textit{off-diagonal} part of $T$. There is no need to regularise these expressions; it can be shown that they are well defined on suitable spaces:
\begin{align}
T_\uod& \phi \uppar n (P,\hat{K}_{n+1})
\label{eq:Toffdiagdef v1} \\
:=
&
  -  \sum_{\substack{i,\ell=1\\i\neq \ell}}^M   \int_{\R^d} \overline{v^\ell_{p_\ell}(-k_{n+1})}  \frac{ v^i_{p_i}(k_{n+1}) \psi \uppar n(P+ (e_i-e_\ell) k_{n+1},\hat{K}_{n+1})}{L(P-e_\ell k_{n+1},K)} \, \ud k_{n+1}  \nonumber
  \\
  &
  - \sum_{i,\ell=1}^M  \sum_{j=1}^{n} \int_{\R^d} \overline{v^\ell_{p_\ell}(-k_{n+1})}  \frac{ v^i_{p_i-\delta_{\ell i} k_{n+1}}(k_j) \psi \uppar n(P-e_\ell k_{n+1}+e_i k_j,\hat{K}_j)}{L(P-e_\ell k_{n+1},K)} \, \ud k_{n+1} \nonumber .
\end{align}
We define for $\nu \in \{1,2\}$ the operator 
\begin{align}
\label{eq:Tdef}
T^\nu \phi \uppar n  := T^\nu_\ud \phi \uppar n  + T_\uod \phi \uppar n 
\end{align} 
sector-wise, by the expressions above, on a domain that will be specified in Proposition~\ref{prop:Tmainprop1} below. 
\begin{rem}
\label{rem:choice}
Clearly, the choice of $T_\ud$ is not unique. There are, however, several possible criteria why to prefer one regularisation over the other. First of all, if $v=\hat{\delta}$ and $\Theta$ and $\omega$ are quadratic, then the theory allows for a convenient intepretation in the position representation. It is most natural to define $T \phi$ as the constant part in an asymptotic expansion of $G \phi$ as $y_{n+1} \rightarrow x_i$. For more details, see \cite[Rem. 3.4]{nelsontype}. In Fourier representation, this choice corresponds to $\nu=1$.

In general, observe that, formally, $T^\nu_\ud$ is equal to the unregularised diagonal part plus $E_{\Lambda=\infty}$, the counter term at infinity. This will be made rigorous in the proof of Proposition~\ref{prop:renorm} below. If $v^i_p=v^i$ are independent of $p$, then choosing $\nu=1$ means that $H^\nu$ can be approximated by a cutoff operator where the sequence of counter terms does not depend on $p$, i.e., is in fact an actual constant. This is the choice that has been made by Nelson and also in \cite{nelsontype}. If the form factors $v^i_p$ do however depend on $p$, then choosing the variant $\nu=2$, as Eckmann did, seems a viable option because $E_\Lambda$ will anyway be an operator. Albeverio has noted in \cite{alb1973} that the counter term used by Eckmann has ``the disadvantage of not having the correct relativistic spectrum of the physical one nucleon energies''. We can make the following observation: On any sector, the operator $T_\ud^{\nu=2}$ is given by a bounded function of $P$. In particular, for $M=1$ the full operator $T^{\nu=2}$ equals zero on the lowest sector (which corresponds to no bosons).
\end{rem}
We will in the next Lemmas prove the main results about the various parts of $T$ and how to approximate them. We remark that the notation for $T^\nu_{\ud, \Lambda}$ differs from the one that has been used in \cite{nelsontype}.
\begin{lem}
\label{lem:Tdiag}
Assume Condition~\ref{cond02} and let $0 \leq \D < \gamma$. Then, for any $\nu \in \{1,2\}$ and any $\eps>0$ small enough, the operators $T_\ud^\nu$ defined in \eqref{eq:Tddef} are symmetric operators on the domain $D(L^{\D/\gamma+\eps})$. Let $T^\nu_{\ud, \Lambda}$ be the same operator with $v^i_p$ replaced by $\chi_\Lambda v^i_p$ and assume Condition~\ref{cond1}. Then $T^\nu_{\ud, \Lambda} \rightarrow  T^\nu_{\ud}$ in norm as operators on $\mathcal{L}(D(L^{\D/\gamma+\eps}),\hilb)$.
\end{lem}
\begin{proof}
We will prove a bound of the form $\norm{(T^\nu_{\ud, \Lambda} -  T^\nu_{\ud}) \psi} \leq f(\Lambda) \norm{L^{\D/\gamma+\eps} \psi}$ for a continuous function $f$ on $[0, \infty)$ which tends to zero as $\Lambda \rightarrow \infty$. This proves convergence. Boundedness follows by setting $\Lambda=0$. The integrals~\eqref{eq:intergalI} and~\eqref{eq:intergalJ} defining $T^\nu_\ud $ are real, so $T^\nu_\ud$ is a real Fourier multiplier. First, let $\nu =2$, define $\xi_\Lambda(q):=1-\chi_\Lambda(q)$ and observe that the action of $T^\nu_{\ud, \Lambda} -  T^\nu_{\ud}$ is given by a sum over $M$ terms of the form 
\begin{align*}
 \int_{\R^d}& \frac{\xi_\Lambda(q) \abs{v_{p}(-q)}^2 }{L(p- q,K,q) }  - \frac{\xi_\Lambda(q) \abs{v_{p}(-q)}^2 }{\Theta(p-q)+\omega(q) }  \, \ud q
\\
&
=  \int_{\R^d}   \frac{-\xi_\Lambda(q)\abs{v_{p}(-q)}^2 \Omega({K})}{L(p- q,K,q)(\Theta(p-q)+\omega(q)) } \, \ud q \, .
\end{align*}
Note that this vanishes for $n=0$. If $\gamma > \D >0$ the absolute value of the integral can, using Lemma~\ref{lem:scaling}, be bounded by
\begin{align*}
\int_{\R^d}   \frac{\xi_\Lambda(q)  \abs{q}^{-2 \alpha} \Omega({K})}{(\abs{p-q}^\gamma+ \abs{q}^\beta + \Omega(K)) \abs{p-q}^\gamma } \, \ud q \leq C \Omega({K})^{\D/\gamma+\delta_\Lambda} \Lambda^{-\beta \delta_\Lambda} 
\end{align*}
with $\delta_\Lambda:= \delta (1-\chi_{[0,1]}(\Lambda))$ and $\delta>0$ small enough. If $\D =0$ however, we estimate the integral for any $\eps \in (0,2)$ by
\begin{align*}
\int_{\R^d}   \frac{\xi_\Lambda(q)  \abs{q}^{-2 \alpha} \Omega({K})}{(\abs{p-q}^\gamma+ \abs{q}^\beta + \Omega(K)) \abs{p-q}^{\gamma(1-\eps/2)} } \, \ud q \leq C \Omega({K})^{\eps/2+\delta_\Lambda} \Lambda^{-\beta \delta_\Lambda} \, .
\end{align*}
choosing $\delta=\eps/2$ small enough, this shows (because $\Omega(K
)\geq 1$) that $T_\ud^{\nu=2}$ is symmetric on $D(L^{\D/\gamma+\eps})$ and that $ T^{\nu=2}_{\ud}-T^{\nu=2}_{\ud, \Lambda} \rightarrow 0$ in norm. According to Condition~\ref{cond02}, for any $\eps >0$ we have
\begin{align*}
 \int_{\R^d} \frac{ \abs{v_{p}(-q)}^2(\Theta(q)-\Theta(p-q))}{(\Theta(p-q)+\omega(q))(\Theta(q)+\omega(q))} \, \ud q 
 &\leq C \left(\left(\abs{p}^{\gamma}\right)^{\D/\gamma+\eps}+1\right) 
\\
&
 \leq C \left(\Theta(p)^{\D/\gamma+\eps}+1\right) \, .
\end{align*}
If we assume Condition~\ref{cond1}, we even have 
\begin{align*}
 \int_{\R^d} \frac{\xi_\Lambda(q) \abs{v_{p}(-q)}^2(\Theta(q)-\Theta(p-q))}{(\Theta(p-q)+\omega(q))(\Theta(q)+\omega(q))} \, \ud q 
 &
 \leq F(\Lambda) \left(\Theta(p)^{\D/\gamma+\eps}+1\right) 
\end{align*}
for some function $F\in C_0[0, \infty)$. This shows the claims for $T_\ud^{\nu=1}$ as well.
\end{proof}
We will now separate two different terms in $T_\uod$, see~\eqref{eq:Toffdiagdef v1}. First, define
\begin{align}
\label{eq:thetadef}
\theta_{i \ell} \psi \uppar n := \int_{\R^d} \overline{v^\ell_{p_\ell}(-k_{n+1})}  \frac{ v^i_{p_i}(k_{n+1}) \psi \uppar n(P+ (e_i-e_\ell) k_{n+1},\hat{K}_{n+1})}{L(P-e_\ell k_{n+1},K)} \, \ud k_{n+1} \, .
\end{align}
Without loss of generality, we will specify to $(i,\ell) =(1,2)$. 
\begin{lem}
\label{lem:DposT1}
Assume $\D\geq 0$. For any $\eps >0$ small enough the operator $\theta_{1 2}$ defined in \eqref{eq:thetadef} is continuous from $D(L^{\D/\gamma+\eps})$ to $\hilb$ and $\theta_{1 2}+\theta_{2 1}$ is symmetric on this domain. Let $\theta_{1 2,\Lambda}$ be the same operator with $v^i_p$ replaced by $\chi_\Lambda v^i_p$. Then $\theta_{1 2,\Lambda} \rightarrow  \theta_{1 2}$ in norm as operators on $\mathcal{L}(D(L^{\D/\gamma+\eps}),\hilb)$.
\end{lem}
\begin{proof}
We will prove convergence and boundedness first by a bound of the form $\norm{ (\theta_{1 2} -  \theta_{1 2,\Lambda}) \psi} \leq f(\Lambda) \norm{L^{\D/\gamma+\eps} \psi}$ for a continuous function $f$ on $[0, \infty)$ which tends to zero as $\Lambda \rightarrow \infty$. This proves convergence. Boundedness follows by setting $\Lambda=0$. Set $\xi_\Lambda(q):=1-\chi_\Lambda(q)$. Then multiply by $\abs{p_2-k_{n+1}}^{2 \frac{(\D+\epsilon)}{2}}$ and its inverse for any $\epsilon>0$, and estimate using the Cauchy-Schwarz inequality
\begin{align*}
&\abs{(\theta_{1 2}-\theta_{1 2,\Lambda}) \psi \uppar n}^2 
\\
&
 \leq 
\int_{\R^d}  \frac{\xi_\Lambda(k)\abs{v^1_{p_1}(k)}^2 \abs{ \psi\uppar n(P + (e_1-e_2) k,\hat{K}_{n+1})}^2 \abs{p_2-k}^{2 (\D+\epsilon)}}{L(P- e_2 k,\hat{K}_{n+1},k)} \, \ud k
\\
& \quad  \times 
\int_{\R^d} \frac{\xi_\Lambda(q)\abs{ v^2_{p_2}(-q)}^2}{L(P- e_2 q,\hat{K}_{n+1},q) \abs{p_2-q}^{2 (\D+\epsilon)}} \, \ud q \, .
\end{align*}
The integral in $q$ can, for $\epsilon$ small enough, be bounded by 
\begin{align*} 
\int_{\R^d} \frac{\abs{ q}^{-2\alpha} \abs{p_2-q}^{-2 (\D+\epsilon)}}{(\abs{p_2 - q}^\gamma +\abs{q}^\beta + \abs{p_1}^\gamma+1)} \, \ud q & \leq C (\abs{p_1}^\gamma+1)^{-(\D+2\epsilon)/\gamma+\delta_\Lambda} \Lambda^{-\beta \delta_\Lambda} 
\\
&
 \leq C \abs{p_1}^{-(\D+2\epsilon)+\gamma \delta} \Lambda^{-\beta \delta_\Lambda}
\, ,
\end{align*}
where we have used Lemma~\ref{lem:scaling}, $\abs{p_1}^\gamma+1 \geq 1$ and the fact that $-(\D+2\epsilon)+\gamma \delta <0$ for $\delta$ small enough. Integrating in $(P, \hat K_{n+1})$ and performing a change of variables $P \rightarrow P+(e_1-e_2) k_{n+1}$ then gives
\begin{align*}
&\int \abs{(\theta_{1 2}-\theta_{1 2,\Lambda}) \psi \uppar n}^2 (P,\hat{K}_{n+1}) \, \ud \hat{K}_{n+1} \ud P
  \\
 &
\leq
C \Lambda^{-\beta \delta_\Lambda} \int \int_{\R^d}  \frac{\xi_\Lambda(k)\abs{v^1_{p_1-k}(k)}^2 \abs{ \psi\uppar n(P,\hat{K}_{n+1})}^2 \abs{p_2}^{2(\D+\epsilon)}}{L(P- e_1 k,\hat{K}_{n+1},k) \abs{p_1-k}^{\D+2 \epsilon-\gamma \delta} } \, \ud k \ud \hat{K}_{n+1} \ud P  \, .
\end{align*} 
In the next step we can safely bound $\xi_\Lambda(k)$ by one, apply Lemma~\ref{lem:scaling} and obtain the upper bound
\begin{align*}
C' \Lambda^{-\beta \delta_\Lambda} \int \abs{ \psi\uppar n(P,\hat{K}_{n+1}) }^2 \abs{p_2}^{2 \D+\gamma \delta} \,  \ud \hat{K}_{n+1} \ud P \, .
\end{align*}
Choosing $\delta= 2 \eps$ proves continuity and convergence because $\abs{p}^\D \leq \Theta{(p)}^{\D/\gamma}$. To prove symmetry, we use a change of variables:
\begin{align*}
&- \langle  \phi, \theta_{1 2} \psi \rangle 
\\
&
=  \int \overline{\phi(P,\hat{K}_{n+1}) v^2_{p_2-k_{n+1}}(k_{n+1})}  \frac{ v^1_{p_1}(k_{n+1}) \psi \uppar n(P+(e_1-e_2) k_{n+1},\hat{K}_{n+1})}{L(P-e_2 k_{n+1},K)} \ud P \ud K
\\
&
=   \int \frac{\overline{\phi(P-(e_1-e_2) k_{n+1},\hat{K}_{n+1}) v^2_{p_2}(k_{n+1})} v^1_{p_1-k_{n+1}}(k_{n+1})}{L(P-e_1 k_{n+1},K)}   \psi \uppar n(P,\hat{K}_{n+1}) \ud P \ud K \, .
%\\
%&
%=
%\langle T_\uod \psi,  \psi \rangle 
\end{align*}
\end{proof}

The remaining parts of $T_\uod$ are sums over terms of the form
\begin{align}
\label{eq:taudef}
\tau_{i \ell}  \psi \uppar n  := \sum_{j=1}^{n} \int_{\R^d} \overline{v^\ell_{p_\ell}(-k_{n+1})}  \frac{ v^i_{p_i-\delta_{\ell i} k_{n+1}}(k_j) \psi \uppar n(P-e_\ell k_{n+1}+e_i k_j,\hat{K}_j)}{L(P-e_\ell k_{n+1},K)} \, \ud k_{n+1} \, .
\end{align}
The domain of the operators $\tau_{i \ell}$ can be characterised in terms of the domain of powers of the operator $\Omega:= \ud \Gamma(\omega)$ alone.
\begin{lem}
\label{lem:T2mainlem}
Assume $\D \geq 0$ and let $u(s) = \frac{\beta }{\gamma} s - \frac{\D}{\gamma}$. Then, for all $s > 0$ such that the following two conditions are satisfied
\begin{align}
\label{eq:cond1s}
u(s) &< 1
\\
\label{eq:cond2s}
0 &< u(u(s)) \, ,
\end{align}
the operators $\tau_{i \ell} + \tau_{\ell i}$ defined in \eqref{eq:taudef} are symmetric on $ D(N^{\max(0,1-s)} \Omega^{s-u(s)+\eps/2})$ for any $\eps >0$ small enough. Let $\tau_{i \ell,\Lambda}$ be the same operator with $v^i_p$ replaced by $\chi_\Lambda v^i_p$. Then $\tau_{i \ell,\Lambda} \rightarrow  \tau_{i \ell}$ in norm as operators on $\mathcal{L}(D(N^{\max(0,1-s)} \Omega^{s-u(s)+\eps/2}),\hilb)$.
\end{lem}
\begin{proof}
We restrict to $n \geq 1$ because $\tau_{i \ell} =0$ for $n=0$.  Denote $\tau = \tau_{i \ell} $ for some $(i,\ell)$.
We will prove a bound of the form $\norm{ (\tau -  \tau_{\Lambda}) \psi} \leq f(\Lambda) \norm{ N^{\max(0,1-s)} \Omega^{s-u(s)+\eps/2} \psi}$ for a continuous function $f$ on $[0, \infty)$ which tends to zero as $\Lambda \rightarrow \infty$. This proves convergence. Boundedness follows by setting $\Lambda=0$.
 Note that, because $\D\geq 0$ and $\beta \leq \gamma$, it holds that $u(s) \leq s$ and therefore the conditions~\eqref{eq:cond1s} and~\eqref{eq:cond2s} already imply that
\begin{align}
\label{eq:cond3s}
u(s),u(u(s)) \in (0,1) \, .
\end{align}
We multiply by $\omega(k_{j})^{\frac{s}{2}}\omega(k_{n+1})^{\frac{s}{2}}$ and its inverse, apply the Cauchy-Schwarz inequality on $\Lz(\R^d \times \lbrace 1, \dots, n \rbrace)$ and use the assumptions on $ v_p^l$ and $\omega$ to obtain
\begin{align*}
&\abs{(\tau-\tau_\Lambda)  \psi \uppar n }^2 
\\
&
\leq
\begin{aligned}[t]
&\sum_{j=1}^{n} \int_{\R^d} \omega(k_{n+1})^s \frac{ \abs{ v^i_{p_i-\delta_{\ell i}  k_{n+1}}(k_{j})}^2  \abs{ \psi\uppar n(P - e_\ell k_{n+1}+ e_i k_j,\hat{K}_j)}^2}{L(P - e_\ell k_{n+1},K) \omega(k_j)^s} \, \ud k_{n+1} 
\\
&
 \times c \sum_{\mu=1}^{n} \omega(k_\mu)^{s} \int_{\R^d}  \frac{(1-\chi_\Lambda(k_j)\chi_\Lambda(q))) \abs{q}^{-2\alpha-\beta s}}{L(P - e_\ell q,\hat K_{n+1},q)}  \, \ud q \,.
\end{aligned} 
\end{align*}
First of all, we have to estimate 
\begin{align*}
(1-\chi_\Lambda(k_j)\chi_\Lambda(q))^2 
&
= (1-\chi_\Lambda(k_j)+\chi_\Lambda(k_j)(1-\chi_\Lambda(q)))^2 
\\
&
\leq (1-\chi_\Lambda(k_j)+1-\chi_\Lambda(q))^2  
\\
&
\leq 2 ( \xi_\Lambda(k_j)+\xi_\Lambda(q))
\end{align*}
Since $u(s) \in (0,1)$, we can apply Lemma~\ref{lem:scaling} to the integral in the second line. We deal separately with the term that does involve a $\xi_\Lambda(q)$ and the one that does not, such that they are bounded by a constant times
\begin{align*}
\xi_\Lambda(k_j)& \Omega(\hat{K}_{n+1})^{-u(s)} +\Lambda^{-\beta \delta_\Lambda} \Omega(\hat{K}_{n+1})^{-u(s)+\delta_\Lambda} 
\\
&
\leq (\xi_\Lambda(k_j)  +\Lambda^{-\beta \delta_\Lambda}) \Omega(\hat{K}_{n+1})^{-u(s)+\delta} \, .
\end{align*}
Here we have used that $\Omega(\hat{K}_{n+1})\geq 1$. In order to deal with the sum over $\mu$, we separate the term $\mu=j$ from the rest and use~\eqref{eq:sum}, giving 
\begin{align*}
\sum_{\mu=1}^{n} \omega(k_\mu)^s  \Omega(\hat{K}_{n+1})^{-u(s)+\delta}
%\leq
% \omega(k_j)^s \omega(k_{j})^{-u(s)}
%  + 
%  \sum_{\substack{\mu=1 \\ \mu \neq j}}^{n} \omega(k_\mu)^s \Omega(\hat{K}_{n+1})^{-u(s)}
%   \\
%&\leq
%\omega(k_{j})^{s-u(s)}
%+ (n-1)^{\max(0,1-s)} \Omega(\hat{K}_{n+1,j})^s \Omega(\hat{K}_{n+1})^{-u(s)}\\
%%
%&\leq \omega(k_{j})^{s-u(s)}
%+ (n-1)^{\max(0,1-s)} \Omega(\hat{K}_{n+1,j})^{s-u(s)}\\
%%
&\leq \omega(k_{j})^{s-u(s)+\delta}
+ (n-1)^{\max(0,1-s)} \Omega(\hat{K}_{j})^{s-u(s)+\delta}
\,.
 \end{align*}
Consequently, we have a bound of the form
\begin{align*}
\abs{(\tau-\tau_\Lambda) \psi \uppar n }^2 \leq  C \abs{(\tau-\tau_\Lambda)\uppar d \psi \uppar n }^2 +C \abs{(\tau-\tau_\Lambda) \uppar {od}  \psi \uppar n }^2\,,  
\end{align*}
 with
\begin{align}\label{eq:tau_d}
\abs{(\tau-\tau_\Lambda) \uppar d \psi \uppar n }^2  
&:= \begin{aligned}[t] \sum_{j=1}^{n} \int_{\R^d} & \frac{(\xi_\Lambda(k_j)  +\Lambda^{-\beta \delta_\Lambda})  \abs{ \psi\uppar n(P - e_\ell k_{n+1}+ e_i k_j,\hat{K}_j)}^2}{L(P - e_\ell k_{n+1},K)}
\\
& \times \frac{  \abs{  v^i_{p_i-\delta_{\ell i}  k_{n+1}}(k_{j})}^2 }{\omega(k_{n+1})^{-s} \omega(k_j)^{u(s)-\delta} } \, \ud k_{n+1} \end{aligned}
\end{align}
and
\begin{align}
\label{eq:tau_od}
&\abs{(\tau-\tau_\Lambda) \uppar {od} \psi \uppar n }^2
\nonumber \\
&:=  n^{\max(0,1-s)} 
\sum_{j=1}^{n} \int_{\R^d}
\begin{aligned}[t]
 &\frac{  (\xi_\Lambda(k_j)  +\Lambda^{-\beta \delta_\Lambda}) \abs{\psi\uppar n(P - e_\ell k_{n+1}+ e_i k_j,\hat{K}_j)}^2}{L(P - e_\ell k_{n+1},K) }\\
 &\times \frac{\abs{  v^i_{p_i-\delta_{\ell i}  k_{n+1}}(k_{j})}^2\Omega(\hat{K}_{j})^{s-u(s)+\delta}}{\omega(k_{n+1})^{-s}  \omega(k_j)^s} \,  \ud k_{n+1}\,.
\end{aligned}
\end{align}
To treat the term~\eqref{eq:tau_d}, we integrate in $(P,\hat K_{n+1})$, perform a change of variables $P \rightarrow P - e_\ell k_{n+1}+ e_i k_j$, and then rename the variables $k_j \leftrightarrow k_{n+1}$. This yields
\begin{align*}
&\int \abs{(\tau-\tau_\Lambda) \uppar d \psi \uppar n (P,\hat{K}_{n+1})}^2 \,  \ud P \ud \hat {K}_{n+1} \\
&
= \sum_{j=1}^{n} \int   \frac{(\xi_\Lambda(k_j)  +\Lambda^{-\beta \delta_\Lambda}) \omega(k_{n+1})^s\abs{ v^i_{p_i- k_j + \delta_{\ell i} k_{n+1}}(k_{j})}^2  \abs{  \psi\uppar n(P,\hat{K}_j)}^2}{\omega(k_j)^{u(s)-\delta}L(P-e_i k_{j},K) } \,  \ud P \ud {K} 
\\
&
= \sum_{j=1}^{n}  \int  \frac{(\xi_\Lambda(k_{n+1})  +\Lambda^{-\beta \delta_\Lambda}) \omega(k_{j})^s \abs{  v^i_{p_i + \delta_{\ell i} k_{j}}(-k_{n+1})}^2  \abs{ \psi\uppar n(P,\hat{K}_{n+1})}^2}{\omega(k_{n+1})^{u(s)-\delta}L(P- e_i k_{n+1},K) } \,  \ud P \ud K  ,
\end{align*}
where, in the last step, we have used the permutation symmetry and our assumption on $v^i_p$. Because we have $u(u(s)) \in (0,1)$ we can choose $\delta$ so small such that also $u(u(s)-\delta) \in (0,1)$. This allows us to apply again Lemma~\ref{lem:scaling} to the $k_{n+1}$-integral in the usual way and to bound it from above by a constant times 
\begin{align*}
\Lambda^{-\beta \delta_\Lambda}& (\Omega(\hat{K}_{n+1})^{-u(u(s)-\delta)+\delta}  +\Omega(\hat{K}_{n+1})^{-u(u(s)-\delta)}) 
\\
&
\leq 2 \Lambda^{-\beta \delta_\Lambda} \Omega(\hat{K}_{n+1})^{-u(u(s)-\delta)+\delta}  \, .
\end{align*}
Therefore, using again the bound~\eqref{eq:sum}, we conclude
\begin{align*}
\int &\abs{(\tau-\tau_\Lambda)  \uppar d \psi \uppar n (P,\hat{K}_{n+1})}^2 \,  \ud P \ud \hat {K}_{n+1}  
\\
&
\leq C \Lambda^{-\beta \delta_\Lambda} n^{\max(0,1-s)} \int  \abs{\Omega^{\frac{s-u(u(s)-\delta)+\delta}{2}}  \psi\uppar n(P,\hat{K}_{n+1})}^2   \ud P \ud \hat{K}_{n+1} \, .
\end{align*}
We proceed similarly with the second term~\eqref{eq:tau_od} and obtain
\begin{align*}
 &
\abs{\tau \uppar {od} \psi \uppar n }^2
\leq
C \Lambda^{-\beta \delta_\Lambda} n^{2 \max(0,1-s)}
\int
\abs{\Omega^{ s-u(s)+\delta}   \psi\uppar n(P,\hat{K}_{n+1})}^2  \ud P \ud \hat{K}_{n+1}  \, .
\end{align*}
This proves the desired bounds for $\delta=\eps/2$, because $u$ is subadditive, $u(s)\leq s$ and thus
\begin{align*}
s-u(u(s)-\delta)+\delta & = s-u(s)+ u(s) -u(u(s)-\delta) + \delta 
\\
&
\leq s-u(s)+ u(s-u(s)+\delta) + \delta  \leq  2(s-u(s)+\delta) \,.
\end{align*}
Symmetry follows also by a change of variables as in Lemma~\ref{lem:DposT1} together with an additional renaming $k_j \leftrightarrow k_{n+1}$ similar to the one we used above.
\end{proof}

\subsection{Proof of Theorem~\ref{thm:main} for $\gamma = \beta$}
The next proposition gives a domain for $T$ as a whole in the case $\gamma=\beta$. For the general case $\beta < \gamma$, the result can be found in Proposition~\ref{prop:Tmainprop2}.

\begin{prop}
\label{prop:Tmainprop1}
Assume Conditions~\ref{Hdefcond}. If $\beta=\gamma$ and $\D< \gamma/2$ then, for any $\eps>0$ small enough and any $\nu \in \{1,2\}$, the operators $T^\nu$ define symmetric operators on the domain $D(T)=D(L^{\D/\gamma+\eps})$. 
\end{prop}
\begin{proof}
We have to deal with $T^\nu_\ud$ and $T_\uod$ separately. The Lemma~\ref{lem:Tdiag} states that $T^\nu_\ud$ defines a symmetric operator on the domain $D(L^{\D/\gamma+\eps})$ for any $\eps>0$ and $\nu \in \{1,2\}$. If $\beta=\gamma$, the function $u(s)$ of Lemma~\ref{lem:T2mainlem} is equal to $s-\D/\gamma$. Therefore the conditions on the parameter $s$ in this lemma reduce to $s \in (2 \D/\gamma,1+\D/\gamma)$. The Lemmas~\ref{lem:T2mainlem} and~\ref{lem:DposT1} taken together combined with the estimate $\Omega \leq L$ then yield that $T_\uod$ is symmetric on $D(N^{1-s} L^{\D/\gamma+\eps/2})$ because 
\begin{align*}
T_\uod& \phi \uppar n
=
  -  \sum_{\substack{i,\ell=1\\i\neq \ell}}^M   \theta_{i \ell} \phi \uppar n
  - \sum_{i,\ell=1}^M  \tau_{i \ell} \phi \uppar n \, .
\end{align*}
We choose $s_\eps = 1 -\eps/2$, which is possible for $\eps$ small enough because $ \D < \gamma/2$. Estimating $N \leq L$ yields the claim. 
\end{proof}

\begin{proof}[Proof of Theorem~\ref{thm:main} for $\gamma = \beta$]
Recall that under the assumption $\gamma=\beta$, Condition~\ref{cond2} reduces to $0 \leq \D < \gamma/3$. Any $\psi \in D(H)$ can be decomposed into $\psi = (1-G)\psi + G \psi$. The first term belongs to $D(L)$ by definition. Corollary~\ref{lem:gonfockspace} shows that $G$ is bounded from $\hilb$ to $D(L^{(1-\D/\gamma)/2 -\eps})$ for any $\eps>0$, so clearly $D(H) \subset D(L^{(1-\D/\gamma)/2 -\eps})$.

Since by Proposition~\ref{prop:Tmainprop1} the operator $T$ is symmetric on $D(L^{\D/\gamma+\eps})$, we conclude that it is symmetric on $D(H)$ as long as $\D<\gamma/3$ (and $\eps$ is chosen appropriately). To prove the self-adjointness, we decompose:
\begin{align}
\label{eq:DposrewriteH v1}
H^\nu
&
= (1-G)^* L (1-G) + T = H_0 + T (1-G) + T G \, .
\end{align}
From \cite[Prop. 2.7]{nelsontype} we know that $H_0 := (1-G)^* L (1-G)$ is self-adjoint and positive. Because the range of $G$ and the domain of $T$ match together we conclude that $TG$ is a bounded operator on $\hilb$. To prove that $T(1-G)$ is relatively bounded by $H_0$, we simply use Young's inequality as in \cite[Sect. 2.3]{nelsontype}. 
\end{proof}
The proof of the Theorem~\ref{thm:main} in the general case is given in Proposition~\ref{prop:proofofmainthm}.

\subsection{Renormalisation}
We will now prove that the operator $H$ can be approximated by a sequence of cutoff Hamiltonians $H_\Lambda+E_\Lambda$. Let us first recall the definition of these cutoff Hamiltonians. Let $V_\Lambda$ be the interaction operator with form factors $v^i_p$ replaced by $v^i_p \chi_\Lambda$, where $\chi_\Lambda$ is the characteristic function of a ball with radius $\Lambda$ (in the variable $k$ only). Since $v^i_p \in \Lz_{\mathrm{loc}}$, the operator $V_\Lambda$ maps into $\Lz(\R^{dM}) \otimes \Lz(\R^{d})$. Thus the operator
\begin{align*}
H_\Lambda 
= L + a(V_\Lambda) + a^*(V_\Lambda)  
\end{align*}
is self-adjoint on $D(H_\Lambda)=D(L)$. Define $G_\Lambda = -L^{-1} a^*(V_\Lambda)$. We can rewrite the cutoff Hamiltonian analogously to $H$ and arrive at
\begin{align*}
H_\Lambda+E_\Lambda^\nu= (1-G_\Lambda)^* L (1-G_\Lambda) + T_\Lambda +E_\Lambda^\nu \,.
 \end{align*} 
Because $V_\Lambda$ is regular, here $T_\Lambda$ is simply the bounded and in particular self-adjoint operator
\begin{equation*}
 T_\Lambda:=a(V_\Lambda) G_\Lambda = - G_\Lambda^* L G_\Lambda \, ,
\end{equation*}
and $E_\Lambda^\nu$ are the counter terms:
\begin{align*}
E^\nu_\Lambda(P) := \begin{cases}  \sum_{i=1}^M \int_{B_\Lambda} \abs{v^i_{p_i-k}(k)}^2 (\Theta(k)+\omega(k))^{-1} \, \ud k & \nu =1 \\
\sum_{i=1}^M \int_{B_\Lambda} \abs{v^i_{p_i-k}(k)}^2 (\Theta(p_i-k)+\omega(k))^{-1} \, \ud k  & \nu =2  \, . \end{cases} 
\end{align*}
The constants are bounded and self-adjoint operators on $\Lz(\R^{dM})$ by Lemma~\ref{lem:scaling}. Going through the computation \eqref{eq:Tformally v1} with $v^i_p$ replaced by $\chi_\Lambda v_p^i$, we observe that a similar decomposition of $T_\Lambda$ into diagonal and off-diagonal terms is possible. Since $T_\Lambda$ has not yet been modified, it would not converge in the limit $\Lambda \rightarrow \infty$, precisely because of the divergence of the integrals that had to be modified in \eqref{eq:Tddef}. This modification, that seemed to be somewhat ad hoc back then, can be achieved by adding the counter terms to the diagonal part and letting $\Lambda$ go to infinity. That is, we can decompose into $T_\Lambda +E_\Lambda^\nu= T^\nu_{\ud, \Lambda}+T_{\uod, \Lambda}$, where $T^\nu_{\ud, \Lambda}$ and $T_{\uod, \Lambda}$ are exactly the operators defined in \eqref{eq:Tdef}, with $v^i_p$ replaced by $\chi_\Lambda v_p^i$. Recall that the notation $T^\nu_{\ud, \Lambda}$ differs from the one used in \cite{nelsontype}. 

We will state the next Proposition in the case where $\beta=\gamma$. The general case is treated in Proposition~\ref{prop:generalconv}.

\begin{prop}
\label{prop:normconv1}
Assume Conditions~\ref{Hdefcond} and~\ref{cond1} and let $\beta=\gamma$. Then $T_\Lambda +E_\Lambda^\nu \rightarrow T^\nu$ in norm as operators in $\mathcal{L}(D(T),\hilb)$.
\end{prop}
\begin{proof}
This follows by decomposing $T_{\uod, \Lambda}$ into $\tau$ and $\theta$-terms, collecting the results of Lemmas~\ref{lem:Tdiag}, \ref{lem:DposT1} and~\ref{lem:T2mainlem} and estimating $\Omega \leq L$.
\end{proof}

\begin{proof}[Proof of Proposition~\ref{prop:renorm}] 
Let us calculate the difference of resolvents:
\begin{align}
(H_\Lambda +& E^\nu_\Lambda +\ui)^{-1} - (H^\nu+\ui)^{-1} 
\nonumber \\
=
&
(H_\Lambda + E^\nu_\Lambda +\ui)^{-1}(H^\nu -(H_\Lambda + E^\nu_\Lambda) )(H^\nu+\ui)^{-1}
\nonumber \\
=
&
(H_\Lambda + E^\nu_\Lambda +\ui)^{-1}(G_\Lambda-G)^* L (1-G) (H^\nu+\ui)^{-1}
\label{line1} \\
&+(H_\Lambda + E^\nu_\Lambda +\ui)^{-1}(1-G_\Lambda)^* L (G_\Lambda-G) (H^\nu+\ui)^{-1}
\label{line2} \\ \label{line3}
&+(H_\Lambda + E^\nu_\Lambda +\ui)^{-1}(T^\nu -(T_\Lambda + E^\nu_\Lambda) )(H^\nu+\ui)^{-1} \, .
\end{align}
Because $L (1-G) (H^\nu+\ui)^{-1}$ is bounded and $G_\Lambda \rightarrow G$ in norm according to Proposition \ref{prop:generalboundong}, the expression \eqref{line1} converges in norm to zero. Clearly, $T_\Lambda +E_\Lambda^\nu$ is relatively bounded by the operator $(1-G_\Lambda)^* L (1-G_\Lambda)$ but more precisely it is bounded with constants inpendent of $\Lambda$.  This implies that $L(1-G_\Lambda)(H_\Lambda + E^\nu_\Lambda +\ui)^{-1}$ is bounded uniformly in $\Lambda$, so the norm of \eqref{line2} goes to zero as well. The convergence of \eqref{line3} follows from Proposition~\ref{prop:normconv1} or Proposition~\ref{prop:generalconv} and the fact that $T^\nu$ is relatively bounded by $H^\nu$ on $D(H)$. 
\end{proof}

\begin{rem}
Of course the most important result of this article is the Theorem~\ref{thm:main} -- which directly characterises the explicit action and the domain of the Hamiltonian. In earlier works (\cite{eckmann1970,DissAndi}) on these models it was proved that the sequence of cutoff Hamiltonians converges to a self-adjoint and bounded from below operator, and Proposition \ref{prop:renorm} shows that we have identified this very limit. Because the old approach did not succeed in identifiying the limit, it is all the more surprising that the estimates, which are needed in \cite{DissAndi}, are so similar to the ones that we have proved. Let us explain. In Eckmann's approach, the resolvent of the cutoff Hamiltonian is expanded in a Neumann series 
\begin{align*}
(H_\Lambda+E_\Lambda-z)^{-1} = (L-z)^{-1} \sum_{n=0}^\infty \left[-(a(V)+a^*(V_\Lambda) + E_\Lambda)(L-z)^{-1}\right]^n \, .
\end{align*} 
In \cite{DissAndi}, where the reordering method due to Eckmann is worked out in detail, it is observed that the terms of the form $a(V_\Lambda)(L-z)^{-1} a^*(V_\Lambda)$ are the ones that do not converge for fixed $z \in \C$. The series is then regrouped in such a way that terms which are of the same order in the form factor $v_p$ are put together. In particular the terms $E_\Lambda$ and $a(V_\Lambda)(L-z)^{-1} a^*(V_\Lambda)$ both are of order two. The crucial step in the proof is then to show that the sum of these two terms is a Cauchy sequence if the occuring suitable powers of the free resolvent $(L-z)^{-1}$ are taken into account. In our language, for $z=0$, this is of course nothing but the fact that $a(V_\Lambda)G_\Lambda+E_\Lambda \xrightarrow{\Lambda \rightarrow \infty} T$ on the domain of some power of $L$, which is the statement of Proposition \ref{prop:normconv1}. In this sense the resolvent approach of Eckmann is more close to the IBC method than, for example, the use of dressing transformations (see also \cite[Sect.~3.4]{nelsontype}).
\end{rem}

\subsection{Regularity of domain vectors}
In this section we will discuss the regularity of vectors in $D(H)$.  We already know that $D(H) = (1-G)^{-1} D(L)$. Of course we also have that $D(\abs {H}^{1/2})=(1-G)^{-1}D(L^{1/2}) 
%\subset D(N^{1/2})
$, such that the form domain is characterised by the abstract boundary condition $\psi-G\psi \in D(L^{1/2})$
%, which is non-trivial because we have assumed $\D \geq 0$
. 
%\end{rem}
%The boundary condition is also the reason why the mapping properties of $(1-G)^{-1}$ are more easy to deduce than those of the Gross transformation. All results can be infered from the mapping properties of $G$ on $\hilb$.

%Proposition~\ref{lem:generalboundong} establishes that if $\abs{\hat v(k)}\leq \abs{k}^{-\alpha}$, then the vectors in the domain of the operator $H$ with interaction $v$ have the regularity of those in $D(L^\eta)$ for all $\eta<\frac{2-D}4=1-\frac{d-2\alpha}4$. 

\begin{cor}\label{cor:reg}
Assume the Conditions~\ref{Hdefcond}. Then for every $0\leq \eta < \frac{1}{2} - \frac{\D}{2\gamma}$ we have
\begin{equation*}
D(\abs{H}^{1/2})\subset D(L^\eta).
%\psi\uppar{n}\in D(L^\eta)\cap \hilb \uppar{n}.
\end{equation*}
\end{cor}
\begin{proof} 
Let $\psi \in D(\abs{H}^{1/2})$. Since $\eta \leq 1/2$ and therefore $(1-G)\psi \in D(L^{1/2}) \subset D(L^\eta)$, we have to show that $G\psi\in D(L^\eta)$. We may apply Proposition~\ref{prop:generalboundong} with $s=0$, since $\eta<\frac{1}{2} - \frac{\D}{2\gamma} = \frac{u(0)+1}{2}$. This yields
 \begin{equation*}
 D(\abs{H}^{1/2}) = (1-G)^{-1}D(L^{1/2}) \subset (1-G)^{-1} D(N) \subset D(N^{1/2}) \xrightarrow{\ G \ } D(L^\eta)
%  \norm{L^\eta G\psi} \leq C \norm{\sqrt{N+1} \psi}
\,.
 \end{equation*}
Here we have used Lemma~\ref{lem:Dposnisleftinv} in the third step.
\end{proof}
In order to prove the next proposition, we have to add the Condition~\ref{cond3} to be able to control $G \psi$ from below.
\begin{prop}\label{prop:non-reg}
Assume Conditions~\ref{Hdefcond} and \ref{cond3}. Then for any $\frac{1}{2} - \frac{\D}{2 \gamma} \leq \eta \leq 1$ we have
 \begin{equation*}
D(\abs{H}^{1/2})\cap D(L^\eta)=\{0\}.
 \end{equation*}
\end{prop}

\begin{proof}
We will show that $G$ maps no $0 \neq \psi \in  \hilb$ into $D(L^{\eta_0})$ where $\eta_0 = \frac{1}{2} - \frac{\D}{2 \gamma}$. This will show that $D(\abs{H}^{1/2}) \cap D(L^{\eta_0}) =\{0 \}$ because $\eta_0 \leq 1/2$ and therefore $(1 - G) \psi \in D(L^{\eta_0})$ for $\psi \in D(\abs{H}^{1/2})$. The claim will then follow immediately due to the fact that for any $\eta \leq 1$ larger than $\eta_0$ it holds that $D(L^\eta) \subset
D(L^{\eta_0})$.

Let $n\in \N$ be such that $\psi\uppar{n}\neq 0$, and let $R>0$. Define $U\subset \R^{dM}\times \R^{d(n+1)}$ to be the set 
 \begin{equation*}
  U=\{(P,K)\vert R>|p_j| \text{ and } R > \abs{k_j} \text{ for all } j>1 \}= \R^d\times B_R(0)^{M-1} \times \R^d \times B_R(0)^{n}\, .
 \end{equation*}
% We will prove that
% \begin{equation*}
%  \int_U \abs{\widehat{L^\eta G\psi\uppar{n}}(P,K)}^2 \, \ud P \ud K = \infty\,,
% \end{equation*}
%which implies that $G \psi\uppar{n}\notin D(L^\eta)$. 
We first use that $(a+b)^2\geq \tfrac12 a^2- b^2$ and obtain the following lower bound:
\begin{align}
 \label{eq:L^sG decomp 1}
 \abs{L^\eta G\psi\uppar{n}(P,K)}^2
 \geq& \frac{1}{2(n+1)} \frac{\abs{ v^1_{p_1}(k_1)}^2 \abs{\psi\uppar{n}(P+e_1 k_1, \hat K_1)}^2}{L(P,K)^{2-2\eta}}\\
 &-  M \sum_{(i,j)\neq (1,1)} \frac{\abs{ v^i_{p_i}(k_j)}^2\abs{ \psi\uppar{n}(P+e_i k_j, \hat K_j)}^2}{L(P,K)^{2-2\eta}}.
 \label{eq:L^sG decomp 2}
\end{align}
We will see that the terms~\eqref{eq:L^sG decomp 2} have a finite integral over $U$, while the integral of~\eqref{eq:L^sG decomp 1} diverges if $R>0$ is chosen large enough.
In the sum over the tupels $(i,j)\neq (1,1)$ in~\eqref{eq:L^sG decomp 2}, have a look at the terms with $i=1, j> 1$. First of all, we may completely drop $L$ in the denominator, because it is clearly bounded from below by one. Using a change of variables $p_1 \rightarrow p_1+k_j$ we obtain the upper bound
\begin{align*}
 \int\limits_U & \frac{\abs{v^1_{p_1}(k_j)}^2\abs{ \psi\uppar{n}(P+e_1 k_j, \hat K_j)}^2}{L(P,K)^{2-2\eta}} \, \ud P \ud K 
 \\
 &
\leq  \int\limits_U \abs{v^1_{p_1-k_j}(k_j)}^2\abs{ \psi\uppar{n}(P, \hat K_j)}^2 \, \ud P \ud K 
\\
&
\leq \int \abs{ \psi\uppar{n}(P, \hat K_j)}^2 \int_{B_R}  \abs{v^1_{p_1}(-k_j)}^2 \ud k_{j} \, \ud P \ud \hat{ K}_j \, .
 %\leq \norm{\psi\uppar{n}}^2 \sup_{p_1 \in \R^d} \int_{\R^d} \chi_{R}(k) \abs{v_{q_1}(-k)}^2 \, \ud k 
\end{align*}
This is finite since $v^1_p(-k)$ is bounded uniformly in $p$ by a function in $\Lz_{\mathrm{loc}}$. The terms with $i,j>1$ can be bounded by enlarging the domain of integration in the variable $p_i$ to $\R^d$. Then we can go on as for $i=1$. The terms where $j=1$ but $i>1$ are estimated in the same way, but the change of variables is performed in $k_1$ and the remaining integral is then over $p_i$. This results in
\begin{align*}
% &\int\limits_U \abs{ v_{p_i}(k_1)}^2\abs{ \psi\uppar{n}(P+e_i k_1, \hat K_1)}^2 \ud P \ud K 
%% \\
%% &
% \leq 
 \int \abs{\psi\uppar{n}(\hat P_i, K)}^2 \int_{B_R} \abs{v^i_{p_i}(k_1-p_i)}^2 \ud p_i \, \ud \hat P_i \ud K 
% \leq \norm{\psi\uppar{n}}^2 \sup_{q_1\in \R^d} \int\limits_{\abs{p_i}<R}\hspace{-6pt} \frac{\abs{\hat v_{p_i}(q_1-p_i)}^2}{\omega(q_1-p_i)^{2-2\eta}} \, \ud p_i 
 \, .
\end{align*}
If we employ the fact that $v^i_{p_i}(k_1-p_i) = v^i_{k_1}(p_i-k_1)$, we can conclude as above.

To bound the integral over the term~\eqref{eq:L^sG decomp 1} from below, we first perform the usual change of variables $p_1 \rightarrow p_1+k_1$. Then we restrict the domain of integration to $\{\abs{p_1}<R\} \cap U$ to bound it by
\begin{align}
\label{eq:stillwithL}
  \int\limits_{B_R(0)^{M+n}} \abs{ \psi\uppar{n}(P, \hat K_1)}^2 \int\limits_{\R^d} \frac{\abs{v^1_{p_1}(-k_1)}^2 }{ L(P-e_1 k_1, K) ^{2-2\eta}} \, \ud k_1 \ud P \ud \hat K_1\, .
\end{align}

Since we have restricted to $(P,\hat{K}_1) \in B_R(0)^{M+n}$,
%$p^2_i < 1 $ for all $i \in \lbrace 1, \dots, M\rbrace$
% and assumed that $\omega \in L^\infty_\mathrm{loc}$, 
 it holds that $\sum_{i=2}^M \Theta(p_i)+\Omega(\hat{K}_1) \leq C$ for some $C > 0$ that depends on $R$. Because in particular $\abs{p_1} < R$,
%and $1\leq \omega(k_1)$
we can then estimate by using Condition~\ref{cond3}
\begin{align*}
 L(p_1-k_1, \hat{P}_1,K) 
 &
 \leq \Theta(k_1-p_1) + \omega(k_1)  + C  
% \\
% &
% \leq \tilde C( \abs{k_1}^{\gamma} +1) + \omega(k_1)  + C
  \leq C' ( \abs{k_1}^{\gamma} + 1) \, ,
\end{align*}
for some $C' > 0$ that depends on $R$. Condition~\ref{cond3} also allows us to bound $\abs{ v^1_{p_1}(-k_1)}^2$ from below by some constant times $(\abs{k_1}^{\alpha}+1)^{-2}$. Hence the integral~\eqref{eq:stillwithL} is bounded from below by some constant times
\begin{align*}
\int\limits_{B_R(0)^{M+n}} \abs{ \psi\uppar{n}(P, \hat K_1)}^2 \, \ud P \ud \hat K_1 \int_{\R^d} \frac{1}{(\abs{k_1}^{\gamma} + 1)^{2-2\eta} (\abs{k_1}^{\alpha}+1)^2} \, \ud k_1  \, .
\end{align*}
Because $\psi \uppar n \neq 0$, we can choose an $R>0$ large enough such that 
\begin{equation*}
 \int\limits_{B_R(0)^{M+n}} \abs{ \psi\uppar{n}(P, \hat K_1)}^2 \, \ud P \ud K > 0\,.
\end{equation*}
But since $(2-2\eta)\gamma+2\alpha \leq d$ by hypothesis, the integral in $k_1$ is infinite, and we have proved the claim. 
\end{proof}

\section{Proof of Corollary~\ref{corol:nelsontypemodels}}
\label{sect:examples}
In this section we are going to apply the results obtained in the previous section to the two models we have been discussing in the introduction. That is, we have to check, that the Conditions~\ref{cond01} -- \ref{cond3} are fulfilled. In this way we will prove the Corollary~\ref{corol:nelsontypemodels}.

Clearly in both models we have $\gamma=\beta=1$ and the form factor does not depend on the specific particle, so we will write $v^i_p = v_p$ throughout this section. In Gross' model $v_p = \omega^{-1/2}$ is independent of $p$, so we may choose $\alpha=1/2$ for the upper bound. 
%One immediately checks that it is possible to choose $\alpha=1/2$ for both the general upper bound as well as for the lower bound in Condition~\ref{cond3}. 
In Eckmann's model, this is less obvious since $v_p(k) = \Theta(p)^{-1/2}   \Theta(p+k)^{-1/2}   \omega(k)^{-1/2}$. However, for any $0 \leq \delta < 1$ it holds that 
\begin{align}
\label{eq:tunedbound}
 \Theta(p)^{-1/2}   \Theta(p+k)^{-1/2} \leq c(\mu) \abs{k}^{-(1-\delta)/2}
\end{align} 
pointwise on $\R^d \times \R^d$. Here $c(\mu)< \infty$ as long as $\mu>0$ such that Condition~\ref{cond01} is still fulfilled if the nucleon mass $\mu$ is positive. To see why \eqref{eq:tunedbound} is true, note that
\begin{align*}
 0 &\leq (\abs{k} \abs p - p^2+\mu^2)^2 = k^2 p^2  -2  \abs{ k} \abs{  p } ( p^2+\mu^2) + (p^2+\mu^2)^2
 \\
 & \leq k^2 p^2  +2  ( k \cdot   p ) ( p^2+\mu^2) + (p^2+\mu^2)^2 \, .
\end{align*}
Adding $\mu^2 k^2 + \mu^2$ on both sides, we obtain 
\begin{align*}
\mu^2 (k^2+1) &\leq (p^2+\mu^2) [(k+p)^2+ \mu^2] + \mu^2
\\
&\leq (p^2+\mu^2) [(k+p)^2+ \mu^2 +1] \leq (p^2+\mu^2) (1+\mu^{-2}) [(k+p)^2+ \mu^2 ] \, .
\end{align*}
As a consequence for any $0\leq \delta <1$ we have
\begin{align*}
 \Theta(p)^{-1/2} \Theta(p+k)^{-1/2} 
&
\leq  (\mu^{-2}+\mu^{-4})^{1/4} (k^2+1)^{-1/4} 
\\
&
\leq  (\mu^{-2}+\mu^{-4})^{1/4} \abs{k}^{-(1-\delta)/2}  \, .
\end{align*}
and hence the claimed inequality. That means, since $\omega(k)^{-1/2} \leq \abs{k}^{1/2}$, that the upper bound on $v_p$ is valid with $\alpha=1-\delta/2$ and hence $\D=\delta$ in Eckmann's model. In the remainder of this section we will of course choose $\delta=0$.

The following lemma is inspired by a bound given by Wünsch, see \cite[4.2]{DissAndi}. In our case it implies that Condition~\ref{cond1} and hence also Condition~\ref{cond02} is fulfilled in both models.
\begin{lem}
\label{lem:wuensch}
Let $\mu, \Lambda \geq 0$, $\omega(k) = \sqrt{k^2+1} $ and $\Theta(p) = \sqrt{p^2+\mu^2} $. For any $\eps>0$ small enough, there exists a constant $C>0$ such that pointwise on $\R^{d}$ we have
\begin{align}
\label{eq:Tdiaggeneral}
 \int_{\R^d} \frac{(1-\chi_\Lambda(q)\abs{q}^{-2\alpha} \abs{\Theta(q)-\Theta(p-q)}}{(\Theta(p-q)+\omega(q))(\Theta(q)+\omega(q))} \, \ud q \leq C \Lambda^{-\eps_\Lambda/2} \left(\abs{p}^{\D+\eps}+1\right) \, .
\end{align}
Here $\eps_\Lambda:=\eps \cdot (1-\chi_{[0,1]}(\Lambda))$.
\end{lem}
\begin{proof}
Choose any $R>0$. We use the reverse triangle inequality to estimate $\abs{\Theta(q)-\Theta(p-q)} \leq \Theta(p)$.  Because $\sqrt{q^2+1} \geq c(\abs{q}^{1/2} +1)$ for some constant $c>0$, we obtain for any $\eps \in (0,1)$ the upper bound 
\begin{align*}
C \Theta(p) \int_{\R^d} \frac{\xi_\Lambda(q) \abs{q}^{-2\alpha-1+\eps/2}}{\abs{p-q}+\abs{q}^{1/2}+1} \, \ud q \, ,
\end{align*}
where we have defined $\xi_\Lambda(q) = 1- \chi_\Lambda(q)$.
By applying Lemma~\ref{lem:scaling} with $\delta = \eps$, this is clearly bounded by a constant times $\Lambda^{-\eps_\Lambda/2}$ as long as $\abs{p}<R$ and $\eps$ is small enough. For larger $\abs{p} \geq R$, we use again the triangle inequality, estimate $\Theta(p-q)+\omega(q) \geq \abs{p-q}+\abs q  \geq \abs p$ and obtain the upper bound
\begin{align*}
C \frac{\Theta(p)}{{\abs p}^{1-\D-\eps}} \int_{\R^d} \frac{\xi_\Lambda(q) \abs{q}^{-2\alpha-1+\eps/2}}{(\abs{p-q}+\abs{q}^{1/2}+1)^{\D+\eps}} \, \ud q \leq C' {\abs p}^{\D+\eps} \Lambda^{-\eps_\Lambda/2} \, .
\end{align*}
Since $R$ was arbitrary, this proves the claim.
\end{proof}
\begin{rem}
Note that Condition~\ref{cond02} can be proved to hold with an improved exponent ${\max(\D,\eps)}$. To do so, decompose the integral for larger $\abs{p}\geq R$ into
\begin{align*}
 \int_{B_{p}}& \frac{\abs{q}^{-2\alpha} \Theta(p)}{(\Theta(p-q)+\omega(q))(\Theta(q)+\omega(q))} \, \ud q 
 \\
 &+   \int_{B_{p}^C} \frac{\abs{q}^{-2\alpha} \Theta(p)}{(\Theta(p-q)+\omega(q))(\Theta(q)+\omega(q))} \, \ud q \, ,
\end{align*}
where $B_p \subset \R^d$ is the ball of of radius $\abs p$ centered at the origin. For the first term, we obtain for any $\eps \geq0$ an upper bound of the form
\begin{align*}
\frac{\Theta(p)}{\abs{p}} \int_{B_{p}} \frac{\abs{q}^{-2\alpha} }{\abs{q}^{1-\eps}} \, \ud q \leq 
%C(R) \abs{p}^{d-1-2 \alpha+\eps} \int_{B_{p}} \frac{1}{\left(\abs{\frac{{q}}{\abs{p}}}\right)^{2\alpha+1-\eps}} \, \ud q =
C(R) \abs{p}^{\D+\eps} \int_{B_{1}} \frac{1}{\abs{x}^{2\alpha+1-\eps}} \, \ud x \, .
\end{align*}
For $\D>0$ the integral converges even for $\eps=0$, such that the term is bounded by a constant times $\abs{p}^{\max(\D,\eps)}$. The integral over the complement, we simply estimate by $\Theta(p) \int_{B^C_{p}} \abs{q}^{-2-2\alpha} \ud q$. Then by a change of variables $q \rightarrow q/\abs{p}$ this is seen to be bounded by a constant times $\abs{p}^{\D}$.
\end{rem}
Now have a look at Condition~\ref{cond3}. It is clear that, for $\abs{p}<R$, the inequality $(p-q)^2 \leq c ( q^2 + 1)$ holds for some $R$-dependent constant, w.l.o.g. $c>1$. Because the square root is increasing, we have $\Theta(p-q)  \leq \sqrt{c (q^2+ 1) +\mu^2}$. The triangle inequality then yields 
\begin{align*}
\sqrt{c (q^2+ 1) +\mu^2} \leq \sqrt{c} \left( \sqrt{ q^2+\mu^2}+ 1 \right) = C(\Theta(q)+1) \, .
\end{align*}
This already shows Condition~\ref{cond3} for Gross' model. In Eckmann's model, this very bound allows us to estimate $v_p(k)$ for $\abs p < R$ from below by some constant times $\Theta(p)^{-1/2} (\abs{k}^\gamma +1)^{-1}$, which shows that, since $\gamma=\alpha=1$, also in this case the Condition~\ref{cond3} is fulfilled. 

In Gross' model in two dimensions the parameter $\D$ is equal to $d-2\alpha-\gamma = 0$ and thus smaller than $\gamma/3 = 1/3$. For Eckmann's model in $d=3$ the same is true. As a consequence, we have checked Codition~\ref{cond2} and thus proved Corollary~\ref{corol:nelsontypemodels}.

\section{Massless Bosons}
\label{sect:outlook}
We would like to conclude by briefly discussing a variant of Eckmann's model where the nucleons are massive but the bosons are massless. That is, we would still have $\Theta(p) = \sqrt{p^2+\mu^2}$ and $\mu>0$ but $\omega(k) = \abs{k}$, such that $v_p(k) = \abs{k}^{-1/2}  \Theta(p)^{-1/2}   \Theta(p+k)^{-1/2}$. While to our knowledge this very model has not yet been considered, the corresponding nonrelativistic massless Nelson model is well known in the literature, see, e.g. \cite{froehlich1974,BDP12,GrWu17}. In the following we will sketch the construction of a Hamiltonian for the massless variant. Although $L$ is still invertible in this case, in turns out to be convenient to introduce a positive parameter $\lambda$ and to define $G_\lambda := (a(V)(L+\lambda)^{-1})^*$. Making use of the resolvent identity, it is easy to show that the domain can be equivalently expressed by $D(H)=(1-G_\lambda)^{-1} D(L)$ and that also
\begin{align*}
H = (1-G_\lambda)^*(L+\lambda)(1-G_\lambda) + T_\lambda - \lambda  \, .
\end{align*}
Here $T_\lambda$ is the regularised version of $a(V) G_\lambda$. Note that the inequality $N\leq L$, which was used frequently in the massive case, does not hold anymore. That makes it absolutely necessary to obtain $n$-independent bounds on $G_\lambda$ and $T_\lambda$. To achieve this, in Lemma~\ref{prop:generalboundong} we have to choose $s=1$, which is not possible for $\D=0$. In Gross' model, the form factor is just $v(k)=\omega(k)^{-1/2} = \abs{k}^{-1/2}$ and as a consequence it is impossible to choose a different $\alpha$. In Eckmann's model however, if we are ready to pay the price of a faster diverging renormalisation counter term, the bound \eqref{eq:tunedbound} allows us to choose $\D=\delta$ for any $\delta \in [0,1)$. Then Lemma~\ref{prop:generalboundong} in particular yields that $G_\lambda$ maps $D(L^\eta)$ into itself for any $\eta < (1-\delta)/2$. It is easy to see that the norm of $G_\lambda$ as an operator on $D(L^\eta)$ goes to zero for $\lambda \rightarrow \infty$. Therefore $1-G_\lambda$ is invertible on this domain if $\lambda$ is chosen large enough. Because we may again set $s=1$ in Lemma \ref{lem:T2mainlem}, the latter together with Lemmas \ref{lem:Tdiag} and \ref{lem:DposT1} yield that $T$ is bounded and symmetric on $D(L^{\eps+\delta})$ for any $\eps >0$. For $\delta+\eps$ small enough and $\lambda$ large enough we can use the invertebility of $(1-G_\lambda)$ on $D(L^{\eps+\delta})$ to obtain the bound
\begin{align*}
\norm{T_\lambda G_\lambda \psi} \leq C \norm{L^{\eps+\delta} \psi} &= C \norm{L^{\eps+\delta} (1-G_\lambda)^{-1} (1-G_\lambda) \psi} 
\\
&
\leq C' \left(\norm{L^{\eps+\delta}  (1-G_\lambda) \psi} + \norm{ (1-G_\lambda) \psi}\right) \, .
\end{align*}
With Young's inequality we conclude that $T_\lambda G_\lambda$ is infinitesimally bounded with respect to $(1-G_\lambda)^*(L+\lambda)(1-G_\lambda)$. The same is true of $T_\lambda (1-G_\lambda)$. Hence we can, in the very same way as in the massive case, prove the self-adjointness of the operator $H$. In the upcoming work \cite{masslessnelson}, this method will be extended to treat the massless nonrelativistic Nelson model, where the analysis is slightly more involved.

\paragraph{Acknowledgments.} 
I would like to thank Jonas Lampart, Stefan Teufel and Roderich Tumulka for helpful discussions and careful reading of the manuscript. The author was supported by the German Research Foundation (DFG) within the Research Training Group 1838 \textit{Spectral Theory and Dynamics of Quantum Systems}.

\clearpage

\begin{appendix}

\section{Appendix}
We will assume the global Condition~\ref{cond01} also throughout the Appendix.
\label{sect:appendix}

\begin{proof}[\textbf{Proof of Lemma~\ref{lem:scaling}}]{~}

Set $f_\Lambda :=  (1- \chi_\Lambda(k)) \abs{k}^{-\nu-\delta\beta }$ and $g_p:=(\abs{p-k}^\gamma + \Omega)^{-r+\delta} \abs{p-k}^{-\sigma}$. By estimating for $\delta< r$ the denominator
\begin{align*}
(\abs{p-k}^\gamma + \abs{k}^\beta + \Omega )^{-r} & = (\abs{p-k}^\gamma + \abs{k}^\beta + \Omega )^{-r+\delta} (\abs{p-k}^\gamma + \abs{k}^\beta + \Omega )^{-\delta}
\\
&
\leq
(\abs{p-k}^\gamma  + \Omega )^{-r+\delta} \abs{k}^{-\beta \delta} \, ,
\end{align*} 
we observe that the integral under consideration is bounded from above by $\int_{\R^d} f_\Lambda g_p \ud k $. The function $g_p$ is vanishing at infinity and is the translated version of a function which is symmetric and decreasing, so clearly its symmetric decreasing rearrangement is just $g_p^* = g_0$. Let us compute the symmetric decreasing rearrangement of $f_\Lambda$ if $m:=\nu+\beta \delta>0$. It holds that the superlevel sets are
\begin{align*}
\{ f_\Lambda(k) > t \} = \{k \in \R^d \vert \Lambda <  \abs{k}  < t^{-1/m} \} = B_{t^{-1/m}} \setminus B_{\Lambda} \, .
\end{align*}
Their volume is $\mathrm{vol}(\{ f_\Lambda(k) > t \}) = \mathrm{vol}(B_1 \subset \R^d) (t^{-d/m} - \Lambda^d)$ and therefore their symmetric decreasing rearrangement is equal to
$\{ f_\Lambda(k) > t \}^* = B_{(t^{-d/m} - \Lambda^d)^{1/d}}$. Recall that
\begin{align*}
f_\Lambda^*(k) = \int_0^\infty \mathbb{I}_{\{ f_\Lambda(k) > t \}^*} \ud t \, ,
\end{align*}
which means that $f_\Lambda^*(k)$ is the solution of $\abs{k}^d = f_\Lambda^*(k)^{-d/m} - \Lambda^d$ which reads $f_\Lambda^*(k) = (\abs{k}^d+\Lambda^d)^{-m/ d}$.

If $\Lambda \in [0,1]$  we choose $\delta=0$. If in addition $\nu=0$, we can estimate $f_\Lambda \leq f_0 = 1$. Then we apply the Hardy-Littlewood inequality to $\int_{\R^d} \sqrt{g_p} \sqrt{g_p} \ud k$ and obtain the upper bound
\begin{align*}
\int_{\R^d}  \frac{\abs{k}^{-\sigma} }{(\abs{k}^\gamma+\Omega)^r} \, \ud k \, .
\end{align*} 
If $\Lambda \in [0,1]$ and $0<\nu$ then $f_\Lambda$ is vanishing at infinity because $\nu+\beta \delta>0$. We still choose $\delta=0$ and recall that $f_\Lambda \leq f_0=f_0^*$. Then apply the inequality to $\int_{\R^d} f_0 g_p \ud k$, which yields
\begin{align*}
\int_{\R^d} f_0 g_p \ud k \leq \int_{\R^d} f_0 g_0 \ud k  =  \int_{\R^d}  \frac{\abs{k}^{-\sigma-\nu}}{(\abs{k}^\gamma+\Omega)^{r}} \, \ud k  \, .
\end{align*} 

If $\Lambda>1$, let $\delta>0$. As a consequence $\nu+\beta \delta>0$. The Hardy-Littlewood inequality yields
\begin{align*}
\int_{\R^d} f_\Lambda g_p \ud k &\leq \int_{\R^d} f_\Lambda^* g_0 \ud k  =  \int_{\R^d} \abs{k}^{-\sigma} \frac{(\abs{k}^d+\Lambda^d)^{-(\nu+\beta \delta)/d}}{(\abs{k}^\gamma+\Omega)^{r-\delta}} \, \ud k \\
&
\leq  \int_{\R^d} \abs{k}^{-\sigma-\nu} \frac{\Lambda^{-\beta \delta}}{(\abs{k}^\gamma+\Omega)^{r-\delta}} \, \ud k
 \, .
\end{align*} 
Putting these bounds together we obtain
\begin{align*}
\int_{\R^d}& \abs{k}^{-\sigma-\nu} \frac{\Lambda^{-\beta \delta_\Lambda}}{(\abs{k}^\gamma+\Omega)^{r-\delta_\Lambda}} \, \ud k 
\\
&
\leq
\Omega^{-r+\delta+(d-\nu-\sigma)/\gamma} \Lambda^{-\beta \delta_\Lambda}  \int_{\R^d}  \frac{\abs{q}^{-\sigma-\nu}}{(\abs{q}^\gamma+1)^{r-\delta}}  \, \ud q \, .
\end{align*}
Here we have performed a change of variables $k \rightarrow k/\Omega^{1/\gamma}=:q$. The remaining integral is finite, and independent of $\Lambda$ and $\Omega$, as long as $\nu+\sigma < d$ and $\gamma(r-\delta)+\nu+\sigma>d$. Because $\gamma r+\nu+\sigma>d$, there certainly exists a $\delta_0 \in (0, r)$ such that this holds for all $0 \leq \delta < \delta_0$.
\end{proof}

\begin{prop}
\label{prop:Tmainprop2}
Assume Conditions~\ref{Hdefcond}. Set $u(s) = \frac{\beta }{\gamma} s - \frac{\D}{\gamma}$. Then for any $\epsilon>0$ small enough, any $\nu \in \{1,2\}$ and all $s > 0$ such that the following two conditions are satisfied
\begin{align*}
u(s) &< 1
\\
0 &< u(u(s)) \, ,
\end{align*}
the operators $T^\nu$ are symmetric on $D_{s,\epsilon}(T)=D((N+1)^{\max(0,1-s)} L^{s-u(s)+\epsilon/2})$.
\end{prop}
\begin{proof}
%The proof will be split. The integral~\eqref{eq:Dposregintergal} defining $T_\ud $ is real, so $T_\ud$ is a real Fourier multiplier and it is sufficent to prove that it maps the domain $D(T)$ to $\hilb$. 
For $\beta=\gamma$ and $\D< \gamma/2$, the proof has already been given in Proposition \ref{prop:Tmainprop1}. So let $\beta < \gamma$. We know that (Lemma~\ref{lem:Tdiag}) $T^\nu_\ud$ defines a symmetric operator on the domain $D(L^{\D/\gamma+\epsilon})$ for any $\epsilon>0$. We also have
\begin{align*}
s-u(s) &= \frac{1}{\gamma}(\gamma-\beta) s + \frac{\D}{\gamma}  > \frac{\D}{\gamma}  \, ,
\end{align*}
which means that $s-u(s) + \epsilon/2 > \D/\gamma +\epsilon$ and thus $  D((N+1)^{\max(0,1-s)} L^{s-u(s)+\epsilon/2}) \subset D(L^{\D/\gamma+\epsilon}) $ if we choose an  $\epsilon > 0$ small enough. Here we have estimated $\Omega \leq L$. Therefore Lemmas~\ref{lem:Tdiag}, \ref{lem:DposT1} and~\ref{lem:T2mainlem} together prove the claim.
\end{proof}

\begin{prop}
Assume the Conditions~\ref{Hdefcond}. Let $H_0 := (1-G)^* L (1-G)$. Then the operators $T^\nu$ are symmetric on $D(H)$ and relatively $H_0$-bounded with relative bound smaller than one.   
\label{prop:proofofmainthm}
\end{prop}
\begin{proof}
Any $\psi \in D(H)$ can be decomposed into $\psi = \psi(1-G) + G \psi$ where the first term belongs to $D(L)$. In Lemma~\ref{lem:TnachG} we will show that, for $(\beta, \D)$ satisfying Condition~\ref{cond2}, it is possible to choose an $s>0$ and a small $\epsilon >0$ in Proposition~\ref{prop:Tmainprop2} such that $D(L^{\delta_1}) \subset D_{s,\epsilon}(T)=D((N+1)^{\max(0,1-s)} L^{s-u(s)+\epsilon/2})$ for some $\delta_1 <1$ and additionally that $G$ maps $D(N^{\delta_2})$ into $D_{s,\epsilon}(T)$ for some $\delta_2 < 1$. Because $D(H) \subset D(N)$, this clearly implies that $D(H) \subset D_{s,\epsilon}(T)$ and thus both $T^\nu$ are symmetric on $D(H)$.

%Since Proposition~\ref{prop:Tmainprop1} states that $T$ is symmetric on $D(\ell^{\D/\gamma+\eps})$, we conclude that $T$ is symmetric on $D(H)$ as long as $D<\gamma/3$ (and $\eps$ is chosen appropriately). 
%To prove  self-adjointness, we decompose as usual $H_0 + T (1-G) + T G$.
%\begin{align}
%\label{eq:DposrewriteH v2}
%H 
%&
%= (1-G)^* \ell (1-G) + T = H_0 + T (1-G) + T G \, .
%\end{align}
%From \cite{nelsontype} we know that $H_0 := (1-G)^* \ell (1-G)$ is self-adjoint. 
Because the range of $G$ and the domain of each $T^\nu $ match together we conclude that $T^\nu G$ is an operator from $D(N^{\delta_2 })$ into $\hilb$. Making use of Lemma~\ref{lem:Dposnisleftinv} we can prove that $T^\nu G$ is relatively $H_0$-bounded. To prove that $T^\nu (1-G)$ is relatively bounded by $H_0$ we simply use Young's inequality (see \cite{nelsontype}). 
\end{proof}

\begin{prop}
\label{prop:generalconv}
Assume Conditions~\ref{Hdefcond} and~\ref{cond1}. Then there exists $s>0$ and $\epsilon>0$ admissible in Lemma~\ref{lem:T2mainlem} such that $T_\Lambda +E_\Lambda^\nu \rightarrow T^\nu$ in norm as operators in $\mathcal{L}(D_{s,\epsilon}(T),\hilb)$.
\end{prop}
\begin{proof}
This follows by decomposing $T_{\uod, \Lambda}$ into $\tau$- and $\theta$-terms and collecting the results of Lemmas~\ref{lem:Tdiag}, \ref{lem:DposT1} and~\ref{lem:T2mainlem} in the same way as in Proposition~\ref{prop:Tmainprop2}.
\end{proof}

\begin{lem}
\label{lem:TnachG}
Assume Conditions~\ref{Hdefcond}. Let $u(s) = \frac{\beta }{\gamma} s - \frac{\D}{\gamma}$ and $D_{s,\epsilon}(T)=D((N+1)^{\max(0,1-s)} L^{s-u(s)+\epsilon/2})$. Then for any $(\beta,\D)$ with $0 \leq \D < \frac{\gamma \beta^2}{\beta^2+ 2 \gamma^2}$, there exists an $s>0$ with $u(s)<1$ and $u(u(s))>0$ and $\delta_1 , \delta_2 \in [0,1)$ such that for any $\epsilon>0$ small enough
\begin{itemize}
\item $D(L^{\delta_1}) \subset D_{s,\epsilon}(T)$.
\item $G$ is continuous from $D(N^{\delta_2})$ to $D_{s,\epsilon}(T)$.
\end{itemize} 
\end{lem}
\begin{proof}
We can again assume $ \beta<\gamma$, since the statement for $\beta = \gamma$ was already proved above. Start by looking at the second part of the statement. Proposition~\ref{prop:generalboundong} states that $G$ maps $D(N^{\max(0,1-\sigma)/2})$ into $D(L^\eta)$
%Taken together they yield 
%\begin{center}
%\begin{tikzpicture}
%  \matrix (m) [matrix of math nodes,row sep=1em,column sep=3em,minimum width=2em]
% { D (N^{\max(0,1-\sigma)/{2}} )  & D(L^{\eta}) &	 
% \\ &  D(L^{s-u(s)}) &  D((N+1)^{-\max(0,1-s)})   \\ };  
%  \path[stealth-]
%  	(m-2-3) edge node [above] {$T$}  (m-2-2)
%    (m-1-2) edge node [above] {$G$} (m-1-1)
%    ;
%\end{tikzpicture}
%\end{center}
for $\sigma$ and $\eta$ that satisfy some conditions. Of course we would like to choose $\eta:=s-u(s)+\epsilon/2$ for an $s>0$ admissible in Proposition~\ref{prop:Tmainprop2} and multiply by $n^{\max(0,1-s)}$ such that
\begin{align*}
D(N^{\max(0,1-\sigma)/{2}+\max(0,1-s)}) \ \xrightarrow{\ G \ } \ D(N^{\max(0,1-s)}L^{s-u(s)+\epsilon/2}) = D_{s,\epsilon}(T) \, .
\end{align*}
%\begin{center}
% \begin{tikzpicture}
%  \matrix (m) [matrix of math nodes,row sep=1em,column sep=3em,minimum width=2em]
% { D\left(N^{\frac{\max(0,1-\sigma)}{2}} \right)  & D(N^{\max(0,1-s)})   \\ };
%  
%  \path[stealth-]
%  	(m-1-2) edge node [above] {$T G $}  (m-1-1)
%   % (m-1-2) edge node [above] {$G$} (m-1-1)
%    ;
%\end{tikzpicture}
%\end{center} 
First we have to show that the choice of $s$ and $\sigma$ we want to make is indeed possible. Afterwards the second part of the statement can be proved by showing that
 \begin{align}
 \label{eq:deltadef1}
 \delta_2 :=  \max(0,1-s)+\frac{\max(0,1-\sigma)}{2} < 1 \, .
 \end{align}
The first part of the statement will follow by estimating $N \leq L$ and the fact that
  \begin{align}
 \label{eq:deltadef2}
 \delta_1 :=  \max(0,1-s)+s-u(s)< 1 \, ,
 \end{align}
 because we may choose $\epsilon $ small enough. We will define a family of pairs of parameters $(s,\sigma) \in (0, \infty) \times [0, \infty)$  that is such that all the following conditions are in fact satisfied:
 \begin{align}
 \label{eq:condlistnoreg}
&s-u(s) + \frac{\sigma-u(\sigma)-1}{2}  < 0 \\
\label{eq:condlistsigma}
&u(\sigma)<1 \\
\label{eq:condlistus1}
&u(s)<1 \\
\label{eq:condlistus2}
&u(u(s))>0
 \end{align}
We set $\eta=s-u(s)+\epsilon/2$ in Proposition~\ref{prop:generalboundong}. This leads to the Condition~\eqref{eq:condlistnoreg} because we may always choose $\epsilon$ as small as necessary. As $\sigma-u(\sigma)$ is increasing for $\beta < \gamma$ and $\D \geq0$, \eqref{eq:condlistnoreg} also implies that 
 \begin{align}
 \label{eq:etasmaller12}
 s-u(s)=\eta < \frac{1+u(0)-0}2 =  \frac12 - \frac{\D}{2 \gamma} \leq \frac{1}{2} \, .
 \end{align} 
 In Proposition~\ref{prop:generalboundong} we had to choose a parameter $\sigma\geq0$ which lead to \eqref{eq:condlistsigma}. The Conditions~\eqref{eq:condlistus1} and~\eqref{eq:condlistus2} are due to Proposition~\ref{prop:Tmainprop2}.

Now we prepare for the definition of our pair $(s,\sigma)$. To do so we set
\begin{align*}
S_1:= \frac{\gamma  + \D}{\beta}   \qquad S_2 := \frac{1-\frac{3}{\gamma} \D}{\gamma-\beta} 
\end{align*}
and note that $u(S_1)=1$ and $S_1 > 1$ because $\beta < \D+\gamma$. Furthermore, using the upper bound on $\beta$ and $\D$, we also have that
\begin{align}
\label{eq:boundonS2step1}
S_2&
=\frac{1-\frac{3}{\gamma} \D}{\gamma-\beta} > \frac{1-\frac{3 \beta^2}{\beta^2+ 2 \gamma^2}}{\gamma-\beta} = \frac{\frac{\beta^2+ 2 \gamma^2-3 \beta^2}{\beta^2+ 2 \gamma^2}}{\gamma-\beta} = 2 \frac{\frac{ \gamma^2- \beta^2}{\beta^2+ 2 \gamma^2}}{\gamma-\beta} = 2 \frac{ \gamma+ \beta}{\beta^2+ 2 \gamma^2} 
\\
\label{eq:boundonS2}
&
\geq 2 \frac{ \gamma+ \beta}{2 \gamma \beta+ 2 \gamma^2} > \frac{1}{\gamma} \, .
\end{align}
We are ready to define a family of pairs $( s_\eps ,  \sigma_\eps  )$ such that they fulfill the conditions \eqref{eq:condlistnoreg} - \eqref{eq:condlistus2} as long as $\eps$ is small enough. So for any $\eps>0$ let
\begin{align*}
( s_\eps ,  \sigma_\eps  )
&:=
% \bigg(\min\lbrace S_1,\frac\gamma2 S_2\rbrace-{\eps} \, , \ \max\left[0,\min\lbrace \gamma S_2- 2 S_1,S_1\rbrace-3 \eps\right] + \eps \bigg)
% \\
% &
% = \begin{cases} \left(S_1-\eps,\min\lbrace \gamma S_2- 2 S_1,S_1\rbrace-2 \eps\right) & 2 S_1 < \gamma S_2 \\
% \left(\frac\gamma2 S_2-\eps, \eps\right) & 2  S_1 \geq \gamma S_2 \end{cases}
%  \\
% &
% = \begin{cases} \left(S_1-\eps,S_1-2 \eps\right) & S_1 < S_2, \quad S_2 \geq \frac{3}{\gamma} S_1  \\ \left(S_1-\frac{\eps}{2},\gamma S_2- 2 S_1-\eps\right) & S_1 < S_2, \quad S_2 <3S_1 \\
% \left(S_2-\frac{\eps}{2}, \eps\right) &  S_1 \geq S_2 \end{cases}
%   \\
% &
  \begin{cases} \left(S_1-\eps,S_1- \eps\right) & \gamma S_2 \in [3 S_1-\eps,\infty]  \\ \left(S_1-\eps,\gamma S_2- 2 S_1\right) & \gamma S_2 \in (2 S_1,3 S_1-\eps) \\
 \left(\frac\gamma2 S_2-\eps, 0\right) &  \gamma S_2 \in (1,2 S_1]\end{cases} \, .
\end{align*}
We have used the Inequality~\eqref{eq:boundonS2}. We can see that in fact ${s}_\eps > 0$ and ${\sigma}_\eps \geq 0$ if $\eps$ is small enough. To prove that \eqref{eq:condlistnoreg} is fulfilled, we start by noting that
\begin{align*}
2  s_\eps +  \sigma_\eps =   \begin{cases}  3 S_1  - 3 \eps < \gamma S_2 - 2 \eps  & \gamma S_2 \in [3 S_1-\eps,\infty]  \\  \gamma S_2 - 2 \eps  & \gamma S_2 \in (2 S_1,3 S_1-\eps) \\
 \gamma S_2-2 \eps  & \gamma S_2 \in (1,2 S_1] \end{cases}
 \,,
\end{align*}
which clearly implies that for $\eps$ small enough $2  s_\eps +  \sigma_\eps <  \gamma S_2$. Using this we can prove that \eqref{eq:condlistnoreg} is satisfied:
\begin{align*}
s_\eps-u(s_\eps)-\frac{1}{2}\left(1-\sigma_\eps+u(\sigma_\eps) \right)
&
= 
\frac{1}{2 \gamma} \left(\gamma-{\beta}\right) (2 s_\eps +\sigma_\eps) +\frac{3 \D}{2 \gamma} - \frac{1}{2} 
\nonumber \\
%<
%\frac{1}{2 } \left(1-\frac{\beta}{2}\right) S_2 +\frac{3 \D}{4} - \frac{1}{2} 
&
<
\frac{1}{2\gamma } \left(\gamma -3 \D\right) + \frac{3}{2 \gamma} \D - \frac{1}{2} = 0 \, .
\end{align*} 
It is clear that we have ${\sigma}_\eps < S_1$. Since $u$ is increasing if $\beta >0$, we conclude that $u({\sigma}_\eps) < u(S_1) = 1$. That means that \eqref{eq:condlistsigma} holds. In exactly the same way we can prove that $u(s_\eps) < 1$, so \eqref{eq:condlistus1} is fulfilled. Now we check that because $\beta > \D$
\begin{align*}
u(u(S_1)) = u(1) = \frac{\beta}{\gamma}- \frac{\D}{\gamma} > 0 \, .
\end{align*}
By using the hypothesis and \eqref{eq:boundonS2step1} we also see that
\begin{align*}
u\left(u\left(\frac\gamma2 S_2\right)\right) 
&
= \frac{\beta^2}{2 \gamma^2} \gamma S_2 -  \left(\gamma+{\beta}\right) \frac{\D}{ \gamma^2} 
\\
&> \frac{\beta^2}{\gamma^2}  \gamma \frac{ \gamma+ \beta}{\beta^2+ 2 \gamma^2} - (\gamma+\beta) \frac{\gamma \beta^2}{\gamma^2(\beta^2+ 2 \gamma^2)} >  0 \, .
\end{align*}
Both estimates together prove that \eqref{eq:condlistus2} holds for any $\eps$ small enough. In order to finally compute  $\delta_2=\max(0,1-{s}_\eps) + \frac{ \max(0,1- \sigma_\eps) }{2}$, note that for $\eps$ small enough we have that still $S_1-\eps > 1$ and therefore
\begin{align*}
\delta_2
&
= \begin{cases} 0 & \gamma S_2 \in [3 S_1-\eps,\infty]  \\
\frac{ \max(0,1- \sigma_\eps) }{2} & \gamma S_2 \in ( 2 S_1,3 S_1- \eps) \\
\max(0,1-{s}_\eps) + \frac{ \max(0,1- \sigma_\eps) }{2} & \gamma S_2 \in (1,2 S_1]\end{cases}
\\
&
= \begin{cases} 0 & \gamma S_2 \in [3 S_1-\eps,\infty]  \\
\frac{ \max(0,1-\gamma S_2+2 S_1) }{2} & \gamma S_2 \in ( 2 S_1,3 S_1- \eps) \\
\max(0,1-\frac{\gamma}{2} S_2 + \eps) + \frac{1}{2} & \gamma S_2 \in (1,2 S_1] \, .\end{cases}
\end{align*}
In the second case it holds that $1-\gamma S_2 \in (1 - 3 S_1 + \eps , 1- 2 S_1)$ which implies $\max(0,1-\gamma S_2+2 S_1)/2 <  1/2$. In the third case we can choose $\eps$ so small that $(-\gamma S_2 +2  \eps)/2 < - 1/2$ and as a consequence $\max(0,1-\frac{\gamma}{2} S_2 + \eps) < 1/2$. These estimates imply that 
\begin{align*}
\delta_2 = \max(0,1-{s}_\eps) + \frac{ \max(0,1- \sigma_\eps) }{2}
< \begin{cases} 0 & \gamma S_2 \in [3 S_1-\eps,\infty]  \\
\frac{ 1}{2} & \gamma S_2 \in ( 2 S_1,3 S_1- \eps) \\
1 & \gamma S_2 \in (1,2 S_1] \, . \end{cases}
\end{align*} 
In order to prove that $\delta_1 < 1$, we have to distinguish only two cases:
\begin{align*}
\delta_1 =s_\eps-u(s_\eps) + \max(0,1-s_\eps)  = \begin{cases} s_\eps-u(s_\eps)  & s_\eps > 1 \\
1-u(s_\eps)  & 0<s_\eps \leq 1 \, . \end{cases}
\end{align*}
If $s_\eps >1$, using Estimate~\eqref{eq:etasmaller12}, we conclude that $\delta_1 < 1/2$. If $0<s_\eps \leq 1 $, note that $u(u(s_\eps))>0$ implies $u(s_\eps) >0$ (see also \eqref{eq:cond3s}) and therefore we have by \eqref{eq:condlistus2} that $\delta_1<1$ in this case.
%. If $\gamma S_2 \leq 2 S_1$ then $\sigma_\eps = 0$ and $\delta_2=1$. Therefore~\eqref{eq:deltadef1} yields
%\begin{align*}
%\max(0,1-s_\eps) < 1 - \frac{\max(0,1-\sigma_\eps)}{2} = \frac{1}{2} \, .
%\end{align*}
%If however $\gamma S_2 > 2 S_1$ then $\delta_1 < 1/2$ and we have again 
%\begin{align*}
%\max(0,1-s_\eps) < 1/2 - \frac{\max(0,1-\sigma_\eps)}{2} \leq \frac{1}{2}
%\end{align*}
%Combining these two inequalities with the Estimate~\eqref{eq:etasmaller12}, we obtain the desired bound $\delta_2 < 1$.
\end{proof}
\end{appendix}

% \bibliographystyle{amsalpha}
% \bibliography{bib_QFT}

\begin{thebibliography}{AGHKH88}

\bibitem[Alb73]{alb1973}
S.~Albeverio, \emph{Scattering theory in a model of quantum fields. I}, J. Math. Phys. \textbf{14} (1973), 1800.

  \bibitem[BDP12]{BDP12}
S.~Bachmann, D.-A.~Deckert, A.~Pizzo, \emph{The mass shell of the Nelson model without cut-offs}, J. Funct. Anal. \textbf{263} (2012), no.~5, 1224--1282.

\bibitem[CDF{\etalchar{+}}15]{Co_etal15}
M.~Correggi, G.~Dell’Antonio, D.~Finco, A.~Michelangeli, and A.~Teta, \emph{A class of {H}amiltonians for a
  three-particle fermionic system at unitarity}, Math. Phys. Anal. Geom.
  \textbf{18} (2015), no.~1, 32.
  
  \bibitem[DP14]{DeckertPizzo14}
D.-A.~Deckert, A.~Pizzo, \emph{Ultraviolet Properties of the Spinless, One-Particle Yukawa Model}, Comm. Math. Phys. \textbf{327} (2014), no.~3, 887--920.
  
  \bibitem[Eck70]{eckmann1970}
J.~P.~Eckmann, \emph{A model with persistent vacuum}, Comm. Math. Phys. \textbf{18} (1970), 247--264.

\bibitem[Fro74]{froehlich1974}
J.~Fröhlich, \emph{Existence of dressed one electron states in a class of persistent models}, Fortschritte der Physik \textbf{2} (1974), 
159--198.


\bibitem[GW16]{GrWu16}
M.~Griesemer and A.~W{\"u}nsch, \emph{Self-adjointness and domain of the
  {F}r{\"o}hlich {H}amiltonian}, J. Math. Phys. \textbf{57} (2016), no.~2,
  021902.

\bibitem[GW18]{GrWu17}
M.~Griesemer and A.~W{\"u}nsch, \emph{On the domain of the {N}elson {H}amiltonian}, J. Math. Phys. \textbf{59} (2018), no.~4,
  042111.
  
  \bibitem[Gro73]{gross1973}
L.~Gross, \emph{The Relativistic Polaron without Cutoffs}, Comm. Math. Phys. \textbf{31} (1973), 25--73.
  
  \bibitem[Hep69]{Heppskript}
K.~Hepp, \emph{Théorie de la renormalisation}, Lecture Notes in Physics Vol. 2.,
  Springer (1969).


\bibitem[Lam18]{La18}
J.~Lampart, \emph{A nonrelativistic quantum field theory with point interactions in three dimensions}, arXiv preprint arXiv:1804.08295 (2018).
  
  \bibitem[LS18]{nelsontype}
J.~Lampart and J.~Schmidt, \emph{On Nelson-type Hamiltonians and abstract boundary conditions}, arXiv preprint arXiv:1803.00872v3 (2018).

\bibitem[LSTT17]{IBCpaper}
J.~Lampart, J.~Schmidt, S.~Teufel, and R.~Tumulka,
  \emph{Particle creation at a point source by means of interior-boundary
  conditions}, Math. Phys. Anal. Geom. \textbf{21} (2018) no. 2, 12. 


\bibitem[MS17]{MoSe17}
T.~Moser and R.~Seiringer, \emph{Stability of a fermionic {N+1}
  particle system with point interactions}, Comm. Math. Phys. \textbf{356}
  (2017), no.~1, 329--355.
  


\bibitem[Mos51]{Mosh51}
M.~Moshinsky, \emph{Boundary Conditions for the Description of Nuclear Reactions}, Physical Review \textbf{81} (1951), 347–352.

\bibitem[Nel64]{nelson1964}
E.~Nelson, \emph{Interaction of nonrelativistic particles with a quantized
  scalar field}, J. Math. Phys. \textbf{5} (1964), no.~9, 1190--1197.


\bibitem[Sch18]{masslessnelson}
J.~Schmidt, \emph{The Massless Nelson Hamiltonian and its Domain}, In preparation.

\bibitem[ST18]{timeasymmetry}
J.~Schmidt, R.~Tumulka, \emph{Complex Charges, Time Reversal Asymmetry, and Interior-Boundary Conditions in Quantum Field Theory}, arXiv preprint arXiv:1810.02173 (2018).

\bibitem[Slo74]{sloan1974}
A.~D.~Sloan, \emph{The polaron without cutoffs in two space dimensions}, J. Math. Phys. \textbf{15} (1974), 190.

\bibitem[TT15]{TeTu15}
S.~Teufel and R.~Tumulka, \emph{New type of {H}amiltonians without
  ultraviolet divergence for quantum field theories}, arXiv preprint
  arXiv:1505.04847 (2015).

\bibitem[Tho84]{thomas1984}
L.~E.~Thomas, \emph{Multiparticle {S}chr\"odinger {H}amiltonians with
  point interactions}, Phys. Rev. D \textbf{30} (1984), 1233--1237.

\bibitem[Wue17]{DissAndi}
A.~Wünsch, \emph{Self-adjointness and domain of a class of generalized Nelson models}, PhD-Thesis, Universität Stuttgart (2017).  http://dx.doi.org/10.18419/opus-9384.

\bibitem[Yaf92]{Yafaev1992} 
D.~R.~Yafaev, \emph{On a zero-range interaction of a quantum particle with
the vacuum}, J. Phys. A: Math. Gen. \textbf{25} (1992), no. 4, 963.



\end{thebibliography}

\newcommand{\etalchar}[1]{$^{#1}$}
\providecommand{\bysame}{\leavevmode\hbox to3em{\hrulefill}\thinspace}
\providecommand{\MR}{\relax\ifhmode\unskip\space\fi MR }
\providecommand{\MRhref}[2]{%
  \href{http://www.ams.org/mathscinet-getitem?mr=#1}{#2}
}
\providecommand{\href}[2]{#2}

\end{document}